\documentclass{article}

\usepackage{microtype}
\usepackage{graphicx}
\usepackage{subcaption}
\usepackage{booktabs} 
\usepackage{enumitem}

\usepackage{adjustbox}
\usepackage{hyperref}

\usepackage{caption}
\usepackage{subcaption}

\usepackage[preprint]{icml2026}

\usepackage{amsmath}
\usepackage{amssymb}
\usepackage{mathtools}
\usepackage{amsthm}
\newcommand{\calX}{\mathcal{X}}
\newcommand{\calY}{\mathcal{Y}}
\newcommand{\dTV}{d_{\mathrm{TV}}}
\newcommand{\diamP}{\mathrm{diam}_{\mathrm{TV}}(\mathcal{P})}
\newcommand{\E}{\mathbb{E}}

\usepackage[disable,colorinlistoftodos, prependcaption,textsize=tiny]{todonotes}

\usepackage{aliascnt}
\usepackage[capitalize,noabbrev]{cleveref}

\theoremstyle{plain}
\newtheorem{theorem}{Theorem}[section]

\newaliascnt{proposition}{theorem}
\newtheorem{proposition}[proposition]{Proposition}
\aliascntresetthe{proposition}

\newaliascnt{lemma}{theorem}
\newtheorem{lemma}[lemma]{Lemma}
\aliascntresetthe{lemma}

\newaliascnt{corollary}{theorem}
\newtheorem{corollary}[corollary]{Corollary}
\aliascntresetthe{corollary}

\theoremstyle{definition}

\newaliascnt{definition}{theorem}
\newtheorem{definition}[definition]{Definition}
\aliascntresetthe{definition}

\newaliascnt{assumption}{theorem}
\newtheorem{assumption}[assumption]{Assumption}
\aliascntresetthe{assumption}

\theoremstyle{remark}

\newaliascnt{remark}{theorem}
\newtheorem{remark}[remark]{Remark}
\aliascntresetthe{remark}

\DeclareMathOperator*{\argmin}{arg\,min}

\icmltitlerunning{Structured Credal Learning}

\begin{document}

\twocolumn[
  \icmltitle{Structured Credal Learning}
  \icmlsetsymbol{equal}{*}

  \begin{icmlauthorlist}
    \icmlauthor{Varun Venkatesh}{lmu,mcml}
    \icmlauthor{Eyke Hüllermeier}{info,mcml}
    \icmlauthor{Bernd Bischl}{lmu,mcml}
    \icmlauthor{Mina Rezaei}{lmu,mcml}
  \end{icmlauthorlist}

  \icmlaffiliation{lmu}{Institute of Statistics, LMU Munich}
  \icmlaffiliation{mcml}{Munich Center for Machine Learning}
  \icmlaffiliation{info}{Institute of Informatics, LMU Munich}

  \icmlcorrespondingauthor{Varun Venkatesh}{varun.venkatesh@campus.lmu.de}
  \icmlcorrespondingauthor{Mina Rezaei}{Mina.Rezaei@stat.uni-muenchen.de}

  \icmlkeywords{Machine Learning, ICML}

  \vskip 0.3in
]
\printAffiliationsAndNotice{} 

\begin{abstract}

Real-world learning tasks often encounter uncertainty due to covariate shift and noisy or inconsistent labels. However, existing robust learning methods merge these effects into a single distributional uncertainty set. In this work, we introduce a novel structured credal learning framework that explicitly separates these two sources. Specifically, we derive geometric bounds on the total variation diameter of structured credal sets and demonstrate how this quantity decomposes into contributions from covariate shift and expected label disagreement. This decomposition reveals a \emph{gating effect}: covariate modulates how much label disagreement contributes to the joint uncertainty such that seemingly benign covariate shifts can substantially increase the effective uncertainty. We also establish finite-sample concentration bounds in a fixed covariate regime and demonstrate that this quantity can be efficiently estimated. Lastly, we show that robust optimization over these structured credal sets reduces to a tractable discrete min–max problem, avoiding ad-hoc robustness parameters. Overall, our approach provides a principled and practical foundation for robust learning under combined covariate and label mechanism ambiguity.
\end{abstract}

\color{red}

\color{black}

\vspace{-5pt}
\section{Introduction} 
\vspace{-3pt}
Predictive uncertainty is crucial for developing reliable and trustworthy machine learning algorithms used to make critical decisions, such as medical diagnosis~\citep{bordes2010label,shalit2017estimating,wang2025udel}, drug discovery~\citep{svensson2025enhancing}, and autonomous driving \citep{chytas2024pooling}. This uncertainty arises from multiple, fundamentally distinct sources: distributional mismatch and inherent target ambiguity~\citep{hullermeier2005learning,gal2016dropout}. Specifically, \textit{covariate shift} occurs when the input distribution differs between training and test time while the conditional labeling mechanism distribution remains unchanged. We refer to \textit{label mechanism ambiguity} for uncertainty in the conditional label distribution arising from measurement noise, annotator variability, or intrinsic class overlap, and it imposes a limit on achievable prediction accuracy~\citep{bordes2010label,natarajan2013learning,gruber2025revisiting,huang2023genkl}. Crucially, these sources demand different operational remedies—distributional adaptation for the former, improved supervision or label reconciliation for the latter.

Most existing methods primarily address these challenges in isolation and typically do not \color{black} decouple marginal and conditional uncertainty in robustness certificates. Domain adaptation \citep{chen25f,tamang2025handling}, importance weighting \citep{li2020robust,sugiyama2007covariate}, invariant learning methods \citep{li2024learning,nguyen2025casual}, and distributionally robust optimization (DRO) \citep{nguyen2024beyond,jeongmulti,zhu2025leveraging,nguyen25a} have shown strong performance in mitigating covariate shift by aligning feature distributions across environments, typically under the assumption of stable labeling mechanisms. In contrast, research on label mechanism ambiguity focuses on robust training or uncertainty-aware prediction under fixed input distributions~\citep{hullermeier2005learning,hullermeier2014learning, v286cai25c,zhou2022ramifications}. Recent works~\citep{li2024noisy,caprio2024credal} attempt to diagnostically decompose prediction error into marginal and conditional components. For example, the credal learning framework of~\citet{caprio2024credal}, grounded in imprecise probability theory~\citep{augustin2014introduction}, replaces single distributions with convex sets of distributions (credal sets) to more faithfully represent uncertainty, limited knowledge, and data variability. However, it does not provide formal predictive robustness guarantees that can explicitly attribute observed errors to covariate shift versus label mechanism ambiguity.

In this paper, we propose a structured credal learning framework that explicitly distinguishes the sources of uncertainty into covariate shift and label mechanism ambiguity. Instead of modeling distributional uncertainty as an unstructured set of joint distributions, we construct credal sets as convex combinations of product measures formed by pairing plausible covariate distributions with plausible labeling mechanisms. 

Our main theoretical contribution is a geometric characterization of the resulting credal set under total variation (TV) distance. We show that its diameter decomposes into \emph{interpretable contributions from covariate shift} and \emph{expected label mechanism ambiguity}, revealing a precise interaction between environmental variation and annotator ambiguity. We further establish concentration bounds in a fixed covariate regime with uncertain labeling mechanisms and prove that \emph{the credal diameter equals the maximum pairwise annotator disagreement rate}. This transforms an abstract robustness parameter into a directly measurable statistic that admits tight finite-sample concentration guarantees. Finally, we show that robust learning over a structured credal set admits a tractable DRO formulation where the worst-case risk is attained at an extreme point, reducing learning to a finite min--max selection. This avoids extrinsic radius tuning while restricting the adversary to semantically meaningful alternatives.

Taken together, our contributions are: 
(i) We introduce a structured credal learning framework that explicitly decomposes distributional uncertainty into covariate shift and label mechanism ambiguity. 
(ii) We derive explicit bounds on the total-variation diameter of structured credal sets and show that it decomposes along the same two sources of uncertainty linked by a gating interaction. This yields interpretable robust generalization bounds directly attributable to each source.
(iii) In the fixed-covariate regime, we prove the credal diameter equals the maximum pairwise disagreement rate — a directly measurable statistic estimable at parametric rate — and establish a minimax lower bound.
(iv) We show that robust optimization over structured credal sets reduces to a finite min–max problem over extreme-points, with an intrinsic, observable robustness radius that requires no extrinsic tuning. 

\vspace{-5pt}
\section{Related work}
\vspace{-5pt}
Our work bridges credal learning theory and DRO.

\textbf{Credal learning Theory} (CLT)~\citep{caprio2024credal} extends statistical learning theory grounded on imprecise probability~\citep{augustin2014introduction}. Recent work leverages this perspective to derive finite-sample generalization guarantees whose tightness is controlled by the total variation diameter $\diamP$ of the credal set~\citep{lohr2025credal,javanmardi2024conformalized}. Existing CLT formulations~\citep{shariatmadar2025generalized,wang2025credal,mubashar2025random,caprio2025credal}, however, treat uncertainty at the level of the joint distribution over $(X,Y)$, thereby conflating heterogeneity in the covariate marginal $P_X$ with ambiguity in the conditional labeling mechanism $P_{Y|X}$. In contrast, we introduce a structured credal construction that factorizes uncertainty across covariate distributions and labeling mechanisms enabling a geometric decomposition of the induced diameter into interpretable covariate-driven and label-driven components.  


\textbf{Distributionally robust optimization} defines uncertainty sets as Wasserstein or $f$-divergence balls around the empirical distribution~\citep{ben2013robust,ai2024not,cai2025diagnosing,zhu2021kernel,ehyaei2024wasserstein}, requiring specification of an extrinsic radius parameter and often producing worst-case distributions with limited semantic interpretability. In contrast, our uncertainty set is the convex hull of coherent generative ``worlds'' (covariate  $\times$ labeling mechanism), leading to an intrinsic, data-driven robustness quantity via the credal diameter decomposition, This yields robustness certificates that are directly attributable to distinct sources of distributional variability and a finite extreme-point reduction of the robust objective.

\vspace{-4pt}
\section{Structured Credal Learning Framework}
\label{sec:framework}
\vspace{-5pt}
\subsection{Preliminary: Credal Learning Theory}
\vspace{-3pt}
\paragraph{Notation.} Let $\mathcal{X}$ and $\mathcal{Y}$ denote the feature and label spaces, equipped with $\sigma$-algebras $\mathcal{A}_{\mathcal{X}}$ 
and $\mathcal{A}_{\mathcal{Y}}$. The joint space $\Omega := \mathcal{X} \times 
\mathcal{Y}$ carries the product $\sigma$-algebra $\mathcal{A}_\Omega := 
\mathcal{A}_{\mathcal{X}} \otimes \mathcal{A}_{\mathcal{Y}}$.


CLT generalizes empirical risk minimization (ERM) by replacing the assumption of a single data-generating distribution with a \emph{credal set}, i.e., a convex set $\mathcal{P}$ of plausible probability distributions over $\mathcal{X}\times\mathcal{Y}$, establishing generalization guarantees across $\mathcal{P}$.
This framework yields a powerful robustness certificate: for an ERM $\hat{h}$ trained on a source $P \in \mathcal{P}$, the generalization capability extends to any target $Q \in \mathcal{P}$ via the following bound:
\begin{equation} \label{eq:credal_bound}
\mathbb{P}\left[
L_Q(\hat{h}) - L_P(h^*) \;\leq\; \varepsilon^{*}(\delta,n,\mathcal{H}) \;+\; \eta
\right]
\;\geq\; 1 - \delta,
\end{equation}
where $L_Q(\hat{h})$ \todo[inline]{Typically defined on a probability measure $Q$} is the expected population risk on the target distribution under the ERM hypothesis $\hat{h}$ and $L_P(h^*)$ is the optimal risk on the source. Here, $\varepsilon^{*}$ encapsulates the learnable statistical error
\footnote{Under the corresponding realizability and hypothesis class $\mathcal{H}$ complexity assumptions, Corollaries 4.2, 4.6, and 4.11 in  \citet{caprio2024credal}  are unified by the formulation in \cref{eq:credal_bound}.  In particular, $\varepsilon^{*}$ recovers the respective error terms
$\epsilon^{*}$, $\epsilon^{**}$, and $\epsilon^{***}$ }, while $\eta$ represents the \textit{robustness penalty} defined by the credal diameter, $\eta = \diamP = \sup_{P,Q \in \mathcal{P}} d_{\mathrm{TV}}(P,Q) = \sup_{A \in \mathcal{A}_\Omega} |P(A) - Q(A)|$. 

However, credal learning theory compresses all uncertainty sources—covariate shift and label mechanism ambiguity alike—into the single scalar $\diamP$. This conflation obscures whether robustness penalties arise from distributional mismatch or intrinsic labeling uncertainty, precluding targeted interventions. Two settings may yield identical credal diameters yet require different remedies: high covariate shift with deterministic labels calls for domain adaptation, while stable features with ambiguous labels demand improved supervision. 



This naturally leads to the notion of a \emph{structured credal set}, which serves as the central object of analysis in the remainder of our work. We formalize a structured uncertainty model that captures variability in both the covariate distribution and the label mechanism ambiguity. Rather than representing uncertainty by a single undifferentiated collection of joint distributions, we consider a collection of plausible distributions that explicitly preserves the distinction between covariate variation and labeling uncertainty.

\subsection{Covariate--Labeling Uncertainty}
\label{subsec:covariate_label}

We model uncertainty in the data-generating process as arising from two conceptually distinct sources: variation in the environment, captured by the feature distribution (covariate shift), and variation in the labeling mechanism (labeling uncertainty). Operationally, this corresponds to a two-stage generative description: first sample $X \sim P_X$, then sample $Y \mid X \sim P_{Y \mid X}(\cdot \mid X)$.

\begin{definition}[Feature and Label Distributions]
\label{def:families}
We define two finite collections of probability measures:
\begin{itemize}
    \item Let $\mathcal{E} \color{black}:= \{P_{X}^{(i)}\}_{i=1}^{N_{X}}$ be a collection of probability distributions on $(\mathcal{X}, \mathcal{A}_{\mathcal{X}})$, representing plausible feature distributions or \emph{environments} for short (e.g., different demographics or imaging equipments).
    \item Let $ \mathcal{L} := \{P_{Y|X}^{(j)}\}_{j=1}^{N_Y}$ be a finite collection of conditional probability distributions from $(\mathcal{X}, \mathcal{A}_{\mathcal{X}})$ to $(\mathcal{Y}, \mathcal{A}_{\mathcal{Y}})$,
representing plausible labeling mechanisms or \emph{labelers} for short (e.g.\ different annotators).
\color{black}
\end{itemize}
\end{definition}

Each combination of an environment and a labeler defines a coherent ``world'' or data-generating process. We formalize this coupling via the joint product distribution.


\begin{definition}[Joint Distribution]\label{def:joint-from-components} \todo[inline]{however sub-probability of $Y|X$ means this is not longer a joint distribution over a probability measure}
For $(i,j) \in [N_X] \times [N_Y]$, define $P_{ij}$ on $(\Omega, \mathcal{A}_\Omega)$ for $ A \in \mathcal{A}_{\mathcal{X}},\, B \in \mathcal{A}_{\mathcal{Y}}$ by
\[
P_{ij}(A \times B) := \int_A P_{Y|X}^{(j)}(B \mid x)\, \mathrm{d}P_X^{(i)}(x),
\]
When densities exist, $p_{ij}(x,y) = p_X^{(i)}(x) \cdot p_{Y|X}^{(j)}(y|x)$. 

\end{definition}

\subsection{Structured Credal Set}

\begin{definition}[Structured Credal Set]
\label{def:structured-credal-set}
The \emph{structured credal set} is defined as the convex hull of the joint
distributions induced by all environment--labeler pairs:
\begin{equation}
\mathcal P
\;:=\;
\operatorname{Conv}\big(\{P_{ij}\}_{i \in [N_X],\, j \in [N_Y]}\big).
\end{equation}
 \end{definition}
\vspace{-7pt}

\begin{remark}[Justification and Construction]
The convex hull in \cref{def:structured-credal-set} admits a simple generative interpretation: if the environment--labeler pair $(i,j)$ is chosen based on a fixed unknown mixing distribution $\pi$ over $[N_X] \times [N_Y]$, the resulting data-generating  distribution is 
\(
P_\pi = \sum_{i,j} \pi_{ij} P_{ij}.
\)
Since $\pi$ is unknown, the set of all such mixtures is exactly the $\operatorname{Conv}\big(\{P_{ij}\}\big)$. This construction imposes no adversarial assumption yet encodes uncertainty about which coherent generative world (or mixtures thereof) may be encountered at deployment. In settings like crowdsourcing and medical imaging, where each annotator/expert defines a labeler while different hospitals or patient populations define environments, the sets $\mathcal{E}$ and $\mathcal{L}$ arise naturally. When the generative process is parametric, these can be constructed by discretizing plausible parameter ranges. 

\end{remark}
\color{black}

\begin{lemma}[Extremal Supremum Property]
\label{lem:extremum}
Let $\mathcal{P}$ be the structured credal set. Then:
\begin{enumerate}[label=(\roman*), noitemsep, topsep=2pt]
    \item For any continuous convex $f: \mathcal{P} \to \mathbb{R}$,
    \[
        \sup_{Q \in \mathcal{P}} f(Q) = \max_{Q \in \mathrm{ex}(\mathcal{P})} f(Q).
    \]
    \item For any continuous, separately convex $f: \mathcal{P}^2 \to \mathbb{R}$,
    \[
        \sup_{(Q_1, Q_2) \in \mathcal{P}^2} f(Q_1, Q_2) = \max_{(Q_1, Q_2) \in \mathrm{ex}(\mathcal{P})^2} f(Q_1, Q_2).
    \]
\end{enumerate}
\end{lemma}

\color{black}

\section{Theoretical Analysis: Geometry of the Structured Credal Set}
\label{sec:theory}


The diameter $\mathrm{diam}_{\mathrm{TV}}(\mathcal{P}) := \sup_{P,Q \in \mathcal{P}} \dTV(P,Q)$ governs robust generalization bounds by Equation \eqref{eq:credal_bound}. We show that this diameter decomposes into interpretable contributions from covariate shift and labeling uncertainty, with a particularly clean form when $N_X = 1$ (~\cref{sec:nx1}).

\subsection{Total Variation Decomposition}
\label{subsec:tv-decomposition}

By \cref{lem:extremum} and the convexity of $\dTV$, the $\diamP$ is determined by the maximal distance between the extreme points in $\mathcal{P}$, i.e., any two distinct data-generating processes $P_{ij}$ and $P_{i'j'}$ within our set. 

In the following, we analyze three cases: (1) pure labeler uncertainty under a fixed covariate (fixed $i$), (2) pure covariate shift under a shared labeler (fixed $j$), and (3) the general case where both sources of uncertainty vary simultaneously. Detailed proofs for all propositions and theorems in this section are provided in the Appendix \ref{proofs:general}

\begin{proposition}[Label Disagreement Distance]
\label{prp:label-disagreement}
Let $P_{ij}, P_{ij'}$ share the same environment marginal $P_X^{(i)}$. Then
\begin{multline*}
\dTV(P_{ij}, P_{ij'}) = \\
\mathbb{E}_{X \sim P_X^{(i)}} \left[ d_{\mathrm{TV}}\big( P_{Y|X}^{(j)}(\cdot \mid X), P_{Y|X}^{(j')}(\cdot \mid X) \big) \right].
\end{multline*}
\end{proposition}

When the environment is fixed, the TV distance between two joint distributions reduces to the expected pointwise disagreement between labelers.



\begin{proposition}[Environment Shift Distance]
\label{prp:covariate-shift}

Let $P_{ij}$ and $P_{i'j}$ be joint distributions sharing the same labeler $P_{Y|X}^{(j)}$ but differing in their environment marginals
$P_X^{(i)}$ and $P_X^{(i')}$. Then
\[
\dTV(P_{ij}, P_{i'j})
\;=\;
\dTV\big(P_X^{(i)}, P_X^{(i')}\big).
\] \todo[inline]{holds with equality if $Y|X$ is a proper probability distribution}
\color{black}
\end{proposition}
\vspace{-5pt}
\paragraph{Interpretation.} Under pure environment shift with a fixed labeler, the joint TV distance exactly equals the distance between the environment marginals. While the data processing inequality \citep{Cover2006} dictates a possible loss of distinguishability when observing only the label space $\mathcal{Y}$, we measure divergence in the fully observable joint space $\Omega$. Because the invariant labeler merely redistributes probability mass across $\mathcal{Y}$ at each $ x $, it does not dilute the inherent discrepancy of the underlying features.


\color{black}

We now combine the results from the fixed-environment and fixed-labeler regimes to bound the distance between two arbitrary worlds $P_{ij}$ and $P_{i'j'}$ where both sources of uncertainty vary.

\begin{theorem}[General Distance Bounds]
\label{prp:general-bounds}
Let $P_{ij}, P_{i'j'} \in \mathcal P$ with $i \neq i'$ and $j \neq j'$. 
Then:
\begin{align*}
&\max_{k \in \{i,i'\}}
\left|
\mathbb{E}_{P_X^{(k)}}\!\left[
\dTV\!\left(P_{Y|X}^{(j)}, P_{Y|X}^{(j')}\right)
\right]
-
\dTV\!\left(P_X^{(i)}, P_X^{(i')}\right)
\right|
\\[4pt]
&\qquad \le
\dTV(P_{ij}, P_{i'j'})
\\[4pt]
&\qquad \le
\dTV\!\left(P_X^{(i)}, P_X^{(i')}\right)
+
\min_{k \in \{i,i'\}}
\mathbb{E}_{P_X^{(k)}}\!\left[
\dTV\!\left(P_{Y|X}^{(j)}, P_{Y|X}^{(j')}\right)
\right].
\end{align*}
\end{theorem}

\paragraph{Interpretation.} \cref{prp:general-bounds} reveals that the joint distributional divergence arises not merely as a sum of covariate shift and labeler disagreement, but through an interaction between them. The bounds decompose the distance between $P_{ij}$ and $P_{i'j'}$ via specific interpolation paths creating a dependency structure we term the \emph{Gating Effect}.
\begin{figure}
    \centering
\includegraphics[width=0.99\linewidth]{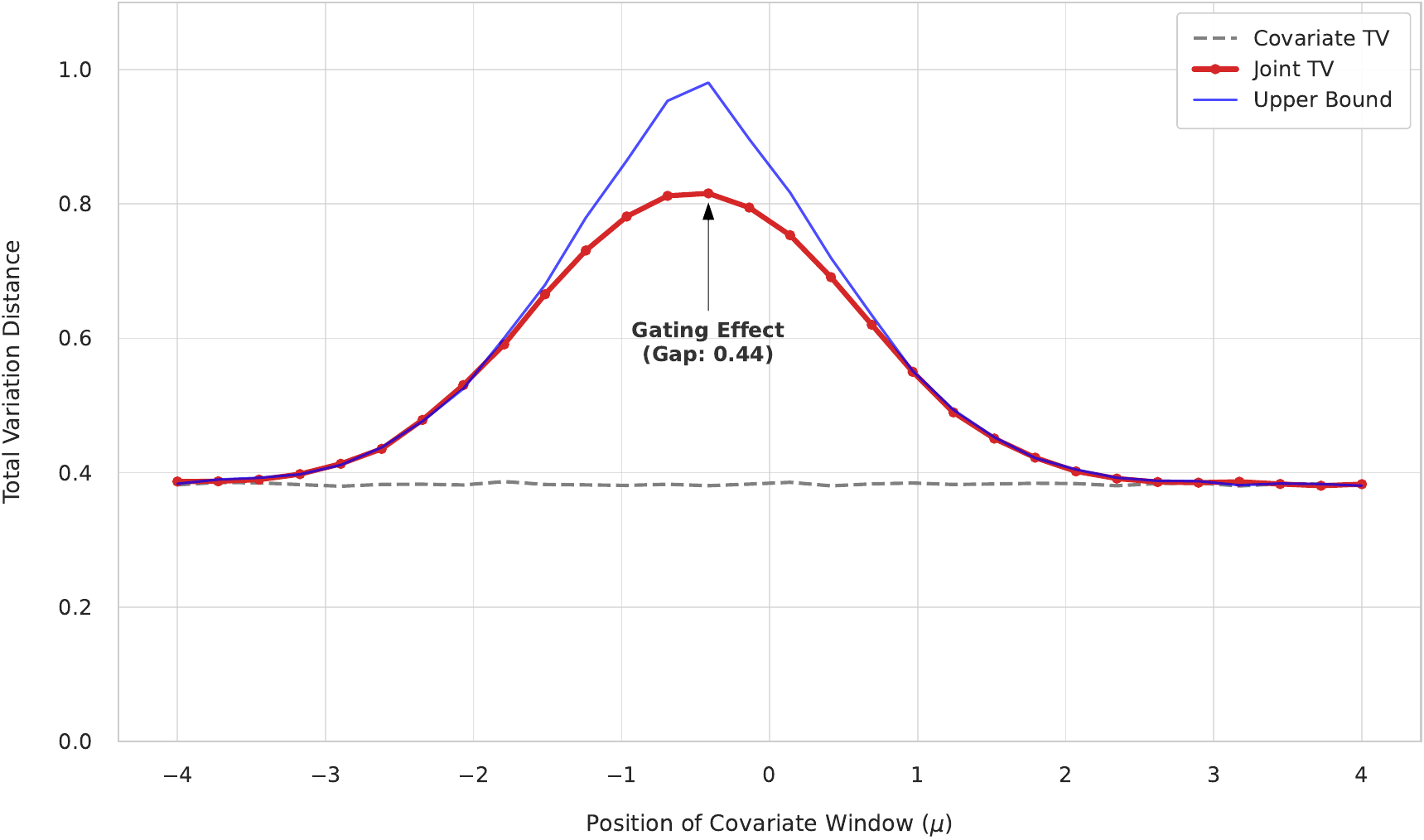}
    \caption{\textbf{The Gating Effect.} Illustration of the interaction between environment and labelers. We simulate two labelers with disjoint decision boundaries (sigmoid thresholds at $x=\pm 1$) and slide a Gaussian covariate window across the input space. While the covariate shift remains constant ($d_{TV}(P_X^{(i)}, P_X^{(i')}) \approx 0.38$, gray dashed line), the \textbf{observable joint distributional distance} (red solid line) spikes significantly when the environment concentrates probability mass in the region where labelers disagree. This demonstrates that the environment acts as a ``gate'', determining the visibility of labeler disagreement; seemingly benign covariate shifts can become catastrophic if they expose previously latent labeler disagreements. The blue line represents the theoretical upper bound derived in  \cref{prp:general-bounds}.}
    \label{fig:gating_effect_1d}
\end{figure}
\vspace{-3pt}
As demonstrated in \cref{fig:gating_effect_1d}, the environment controls the observability of the extent of disagreement. $\mathbb E_{P_X^{(k)}} [ \dTV(P_{Y|X}^{(j)}, P_{Y|X}^{(j')}) ]$ implies that labelers may disagree significantly in region $R \in \mathcal{X}$, yet this disagreement would be visible in the joint space only if the environment places sufficient mass on $R$. 

\textbf{High-Dimensional Gating.} Appendix \cref{fig:app:gating:3d} extends this to $\mathbb{R}^3$ visualizing a trajectory through a disagreement nebula. Even small, constant covariate shifts can induce sharp spikes in joint distributional divergence, indicating that high-dimensional safety is a narrow corridor rather than a continuous region. 

Consequently, training under one environment may therefore underestimate the effective impact of labeler disagreement if deployment occurs in an environment that concentrates mass in regions of high disagreement. The structured credal framework makes this coupling explicit, enabling principled reasoning about generalization under simultaneous covariate shift and labeling uncertainty by explicitly taking combinations over all pairs when bounding the credal-diameter.

\subsection{Diameter Bounds for Product Credal Sets}
\label{subsec:diameter}

\begin{definition}[Component Diameters]
\label{def:component-diameters}
For a structured credal set $\mathcal{P}$, we define: 
\begin{align*}
    &\eta_X := \max_{i, i' \in [N_X]} d_{\mathrm{TV}}\big(P_X^{(i)}, P_X^{(i')}\big), \\
    &\bar{\eta}_{Y|X} := \sup_{x} \max_{j, j' \in [N_Y]} d_{\mathrm{TV}}\big(P_{Y|X}^{(j)}(\cdot|x), P_{Y|X}^{(j')}(\cdot|x)\big) \\ 
    &\eta_{Y|X}^* := \max_{ \substack{i \\ \in [N_X]}} \max_{\substack{j, j' \\ \in [N_Y]}} \mathbb{E}_{P_X^{(i)}} \left[ d_{\mathrm{TV}}\big(P_{Y|X}^{(j)}(\cdot|X), P_{Y|X}^{(j')}(\cdot|X)\big) \right]
\end{align*}
\end{definition}

The $\diamP$ is the key quantity governing the looseness of robust generalization bounds (\cref{eq:credal_bound}). Using the decompositions in \cref{subsec:tv-decomposition}, we can bound the global diameter of $\mathcal{P}$ via \cref{def:component-diameters}. These quantities capture distinct aspects of uncertainty: $\eta_X$ measures the maximal divergence in covariate distributions, while $\{ \eta_{Y|X}^\ast, \bar{\eta}_{Y|X} \}$ capture labeler disagreement at differing levels of aggregation. The interplay between these quantities determines the geometry of the credal set.

\begin{theorem}[Diameter Decomposition]
\label{prp:diameter-decomposition}
Let $\mathcal{P} = \mathrm{Conv}(\{P_{ij}\})$ be the structured credal set, then:
\todo[inline]{Lower Bound $\eta_X$ holds under proper probability distribution}
\begin{equation*}
\label{eq:diameter-bounds}
\max\{\eta_X, \eta_{Y|X}^*\} 
\le 
\mathrm{diam}_{\mathrm{TV}}(\mathcal{P}) 
\le 
\eta_X + \eta_{Y|X}^{\mathrm{eff}} \, ,
\end{equation*}
$ \text{where } \eta_{Y|X}^{\mathrm{eff}} := \min\{\eta_{Y|X}^*,\, (1 - \eta_X)\bar{\eta}_{Y|X}\}$.
\end{theorem}
\textbf{Interpretation.~}
The \emph{lower bound} $\max\{\eta_X, \eta_{Y|X}^*\}$ establishes that distributional ambiguity cannot be reduced below the dominant source of uncertainty---whether from covariate shift or labeling uncertainty. The \emph{upper bound} through the 
effective labeler disagreement, $\eta_{Y|X}^{\mathrm{eff}}$, shows that the diameter of the credal set is governed by an interaction between environments and labelers, operating through two mechanisms:

\textbf{(i) Compounding regime} $(\eta_{Y|X}^{\mathrm{eff}} = \eta_{Y|X}^*)$: 
Covariate shift and labeling uncertainty compound additively. The environment \emph{gates in} disagreement across most of the feature space, allowing labeler disagreement to contribute its full expected magnitude.

\textbf{(ii) Gating regime} $(\eta_{Y|X}^{\mathrm{eff}} = (1-\eta_X)\bar{\eta}_{Y|X})$: 
Since feature distributions overlap in a region of total
mass $1-\eta_X$ (in TV for the worst-case pair), labeling uncertainty can contribute to the credal-diameter only on this overlapping mass, whereby the environment \emph{gates out} labeler disagreement. Outside the overlap, the uncertainty is already saturated by covariate shift.

Which term is active is determined by $\eta_{Y|X}^{\mathrm{eff}}=\min\{\eta_{Y|X}^*,(1-\eta_X)\bar{\eta}_{Y|X}\}$, i.e. whether the expected disagreement or the overlap-limited supremum disagreement is the binding constraint. 
 
\begin{remark}[Tightness]
The lower bound is attained when the diameter-achieving pair $(P_{ij}, P_{i'j'})$ 
has either $i = i'$ (pure labeling uncertainty) or $j = j'$ (pure covariate shift). 
The upper bound is tight when covariate distributions have disjoint support 
($\eta_X = 1$) or when label disagreement is constant across $\mathcal{X}$.
\end{remark}

\textbf{Connection to Generalization.}
Note that \cref{prp:diameter-decomposition} yields a \emph{factored generalization bound} in Equation \ (\ref{eq:credal_bound}):
\todo[inline]{Holds under proper probability measure}
{\small
\begin{equation} \label{eq:factored-genralization-bound}
\mathbb{P}\!\left[
L_Q(\hat{h}) - L_P(h^*) \;\leq\; \varepsilon^{*}(\delta,n,\mathcal{H}) \;+\; \eta_X + \eta_{Y|X}^{\mathrm{eff}}
\right]
\;\geq\; 1 - \delta \, ,
\end{equation}
}
decomposing the robustness penalty ($\eta = \diamP$) into interpretable components ($\eta_X, \eta_{Y|X}^{\mathrm{eff}}$) allowing us to  interpret the ``safety budget'' allocated to each source of uncertainty. 

While the general structured credal framework allows for simultaneous variations in environment and labelers, a critical sub-regime arises when the environment is stable, but the labeling mechanism is uncertain. This corresponds to the setting of $N_X=1$, which we term the \emph{Pure Labeling Uncertainty} regime.

\section{The Pure Labeling Uncertainty Regime} \label{sec:nx1}


%

This regime deserves separate consideration for three reasons. First, it directly models practical scenarios such as \textbf{crowdsourcing}, where a fixed dataset is labeled by multiple annotators with varying expertise \citep{dawid1979maximum, JMLR:v11:raykar10a}, and \textbf{inter-observer variability}, common in medical imaging where experts may grade the underlying pathology differently \citep{mchugh2012interrater, mandrekar2011measures}.  Second, the covariate shift term vanishes $(\eta_X = 0)$, collapsing the bounded quantities in \cref{prp:diameter-decomposition} -- previously an abstract measure-theoretic quantity to a \emph{directly measurable disagreement statistic} with an exact characterization. Third, this exactness enables finite-sample estimation and minimax optimality results providing a complete operational pipeline from data to robustness certificates.

\textbf{Setting.}
Our analysis relies on the following structural assumptions formalized in Appendix \ref{proof:concentration}. We observe $n$ i.i.d.\ inputs $X_1,\dots,X_n \sim P_X$ and, for each input $X_i$,
labels from $N_Y$ labeling mechanisms $\{P^{(k)}_{Y\mid X}\}_{k=1}^{N_Y}$. 
Given $X_i$, the labelers (deterministic or stochastic) are mutually independent and generate either (i) soft labels
$\mathbf p_k(X_i)=P^{(k)}_{Y\mid X}(\cdot\mid X_i)$, or (ii) hard labels
$Y_i^{(k)}\sim P^{(k)}_{Y\mid X}(\cdot\mid X_i)$ resulting in independent tuples $(X_i, Y_i^{(1)},\dots,Y_i^{(N_Y)})$ across samples. 

These assumptions capture the minimal structure of independent multi-annotator labeling on a fixed population, as in crowdsourcing or expert annotation, without assuming a latent ground truth or parametric noise model. They are required only for estimation and concentration, while the geometric characterization of the credal set and its diameter holds independently of these assumptions. Detailed proofs are provided in Appendix \ref{proof:pure_label_main}

\subsection{Geometric Characterization: Diameter as Disagreement}
\begin{proposition}[Geometric Exactness of the Credal Diameter]
\label{prop:geometry}
 \vspace{-5pt}
Let $\mathcal{P} = \mathrm{Conv}(\{P_{1j}\}_{j=1}^{N_Y})$ be a structured credal set with fixed feature distribution $P_X$. The total-variation diameter is:
 \vspace{-5pt}
\begin{align*}
    \mathrm{diam}_{\mathrm{TV}}(\mathcal{P}) &\;=\; \eta_{Y|X}^* \\ &\;=\; \max_{j, j' \in [N_Y]} \mathbb{E}_{X \sim P_X} \left[ d_{\mathrm{TV}}\big(P_{Y|X}^{(j)}, P_{Y|X}^{(j')}\big) \right].
\end{align*}\end{proposition}
\vspace{-7pt}
This diameter $\eta_{Y|X}^*$ bridges measure theory and practice, admitting concrete operational forms depending on the observation regime:
 \vspace{-3pt}
\begin{enumerate}[leftmargin=*, noitemsep, topsep=2pt]
    \item \textbf{Deterministic Hard Labels (Exact):} If the labelers are deterministic, the diameter is the \emph{maximum pairwise disagreement rate}:
    \begin{equation} \label{eq:diameter-is-disagreement}
         \eta_{Y|X}^* = \max_{j, j'} \mathbb{P}\left( f^{(j)}(X) \neq f^{(j')} (X) \right)\Big.
    \end{equation}
   
    \vspace{-5pt}
    \item \textbf{Soft Label Regime (Exact):} If probability vectors $\mathbf{p}^{(j)}$ (e.g. confidence scores) are observed, the diameter is proportional to the expected $L_1$ distance:
    \begin{equation} \label{eq:diameter-is-l1}
        \eta_{Y|X}^* = \tfrac{1}{2} \max_{j, j'} \mathbb{E}_{X}\left[ \| \mathbf{p}^{(j)}(X) - \mathbf{p}^{(j')}(X) \|_1 \right].
    \end{equation}
    
    \vspace{-5pt}
    \item \textbf{Stochastic Hard Labels (Conservative Bound):} If only realizations $y^{(j)}$ are observed, 
    \begin{equation} \label{eq:conservate-stochastic-hard-labels}
    \eta_{Y|X}^* \;\le\; \max_{j, j'} \mathbb{P}\left( y^{(j)} \neq y^{(j')} \right).
    \end{equation}
    This ensures that using empirical disagreement to calibrate robustness is a \emph{safe approximation} that never underestimates the necessary safety budget. See \cref{rem:st-hard-lab}
   \item \textbf{Noisy Labels (Aleatoric Uncertainty):}  Under the binary symmetric noisy-annotator model with latent ground truth $y^*$ and annotator error rates $\varepsilon_j \in [0,0.5]$, the diameter of the resulting structured credal set admits a closed form:
    \begin{equation} \label{eq:diam-is-noise}
    \eta_{Y|X}^* 
    = \max_{j,j'} \mathbb{P}(y^{(j)} \neq y^{(j')})
    = \max_{j,j'} \big( \varepsilon_j + \varepsilon_{j'} - 2 \varepsilon_j \varepsilon_{j'} \big).
    \end{equation}

   Consequently, if all annotators satisfy $\varepsilon_j \le \varepsilon_{\max}$, the induced robustness radius is bounded by
   \begin{equation}
       \eta_{Y|X}^* \le 2\varepsilon_{\max} - 2\varepsilon_{\max}^2,
   \end{equation}
   ensuring non-vacuous robustness even under weak supervision.
\end{enumerate}

\subsection{The Observable Radius: From Hyperparameter to Statistical Estimation}
A key advantage of the structured credal framework in this regime is that the diameter of the uncertainty set is \emph{observable} and \emph{estimable} \color{black}\footnote{Note: We use \emph{observable} to refer to sample-level quantities directly recorded (e.g., $X_{1:n}$ and $y_{1:n}^{(j)}$) or trivially computed (e.g., disagreements), and \emph{estimable}  for population quantities that can be estimated from these with finite-sample guarantees. In particular, $\eta_{Y|X}^*$ is a population diameter, while $\hat\eta_{Y|X}^*$ is its empirical estimator} directly from finite data.
\begin{theorem}[Finite-Sample Concentration]
\label{thm:concentration}
Let $\hat{\eta}_{Y|X}^*$ be the empirical estimator of the diameter computed on a sample of size $n$. For any $\delta \in (0,1)$, with probability at least $1 - \delta$:
\begin{equation}
\label{eq:concentration}
\left|\hat{\eta}_{Y|X}^* - \eta_{Y|X}^*\right| \leq \sqrt{\frac{\ln\big(N_Y(N_Y - 1)/\delta\big)}{2n}}.
\end{equation}
\end{theorem}
\vspace{-7pt}
\textbf{Interpretation and Complexity.}
The bound establishes that the credal diameter can be estimated at the parametric rate $O(n^{-1/2})$ with only logarithmic dependence on the number of labelers $N_Y$. This logarithmic scaling is particularly favorable for crowdsourcing, where adding annotators improves coverage of the label space while incurring negligible sample complexity cost. For instance, with $N_Y=10$ annotators, estimating the diameter to within $\epsilon=0.1$ precision (95\% confidence) requires only $n \approx 375$ samples.

\textbf{Robustness to Assumptions.}
This estimation procedure: (i) does not specify error structures (e.g., confusion matrices); (ii) does not require ground truth labels; and (iii) assumes only conditional independence of annotators. 

\begin{remark}[Remark on Stochastic Hard Labels] \label{rem:st-hard-lab}
In the stochastic hard label regime, the true TV diameter is non-identifiable as the underlying conditional distribution of the labeler is non observable. However, since the TV distance is the infimum over all couplings, the observed \emph{independent coupling} provides a mathematically valid upper bound. In this setting, the concentration bound in Eq.~\eqref{eq:concentration} applies to the \emph{observed pairwise disagreement rate}. Since this upper-bounds the true diameter (\cref{eq:conservate-stochastic-hard-labels}), the finite-sample guarantee ensures the reliable estimation of a \emph{conservative safety certificate}, effectively acting as a valid upper bound on the worst-case risk.
\end{remark}

\subsection{Robust Generalization and Optimality}
Finally, we connect the geometric diameter to learning guarantees. The following result establishes that $\eta_{Y|X}^*$ acts as a safety certificate, bounding the risk transfer between any two plausible interpretations of the data.
\begin{corollary}[Robust Generalization]
\label{corr:robust_generalization}
Let $\hat{h}$ be the empirical risk minimizer on a source distribution $P \in \mathcal{P}$. For any target $Q \in \mathcal{P}$ (representing an unknown true labeling mechanism), with probability $1-\delta$:
\begin{equation}
\label{eq:generalization-bound}
\footnotesize
L_Q(\hat{h}) - L_P(h^*) \;\leq\; \underbrace{\varepsilon^*(\delta, n, \mathcal{H})}_{\text{statistical error}} \;+\; \underbrace{\eta_{Y|X}^*}_{\substack{\text{labeling uncertainty} \\ \text{(irreducible constant)}}}.
\end{equation}
\end{corollary}

This decomposition cleanly separates two fundamentally different sources of error. The term $\varepsilon^*$ captures the difficulty of learning the \emph{source} distribution, vanishing as $n \to \infty$. In contrast, $\eta_{Y|X}^*$ captures the fundamental lack of consensus among labelers. It depends solely on the credal geometry and cannot be reduced by collecting more data from the same ambiguous process. This distinction offers a \emph{nuanced view of error}: a loss of magnitude $\eta_{Y|X}^*$ is not necessarily a model failure, but may be a reflection of genuine task ambiguity. Moreover, recent empirical studies with dense annotation consistently report substantial inter-annotator disagreement: in natural language inference, at least 40\% disagreement is common for ambiguous examples~\citep{jiang2022investigating}, and majority-label instability reaches up to 31\% when increasing annotation density~\citep{nie2020can}. In vision benchmarks, per-item agreement can drop to as low as 4\%, corresponding to disagreement rates approaching 96\% in extreme cases~\citep{schmarje2022one}. 

Is this ambiguity penalty tight? We show that no algorithm can circumvent this cost in the worst case.

\begin{theorem}[Minimax lower bound]
\label{thm:minimax}
Fix $0$--$1$ loss and assume $\mathcal H$ contains all measurable classifiers.
For any $\eta\in(0,1)$, there exists a fixed covariate structured credal set
$\mathcal P=\operatorname{Conv}(\{P_{1k}\}_{k=1}^{N_Y})$ with
$\eta^*_{Y|X}=\diamP$
such that for every $n\in\mathbb N$,
\begin{equation}
\inf_{\mathcal A}\ \sup_{P,Q\in\mathcal P}\
\mathbb E_{\mathcal D_n\sim P^{\otimes n}}
\Big[\,L_Q(\mathcal A(\mathcal D_n)) - L_P^*\,\Big]
\ \ge\ \frac{1}{2}\,\eta^*_{Y|X}.
\end{equation}
Consequently, any upper bound with linear dependence on $\eta^*_{Y|X}$ is unimprovable in the worst case.
\end{theorem}
\vspace{-7pt}
Theorem~\ref{thm:minimax} confirms that $\eta_{Y|X}^*$ represents the \emph{price of ambiguity}. To build intuition, consider two labelers with expected disagreement $\eta_{Y|X}^*$. On the region where they disagree, any hypothesis agreeing with one disagrees with the other. Any stochastic strategy incurs expected error at least $0.5$ on the disagreement region.

\textbf{Summary: A Three-Layer Unity.} 
Together, \cref{sec:nx1}  demonstrates that in the pure labeling uncertainty regime, the credal set is \emph{fully observable}: its diameter is not a hyperparameter to be tuned but a statistic to be measured.  $\eta_{Y|X}^*$  achieves a significant confluence in learning theory and plays a coherent role at three distinct levels of analysis.
(i) \textbf{Geometrically}, it is the \emph{diameter} of the uncertainty set (in TV distance).
(ii) \textbf{Statistically}, it is an \emph{observable parameter} (max pairwise disagreement) estimable at rate $O(n^{-1/2})$.
(iii) \textbf{Theoretically}, it is the \emph{minimax price of ambiguity}, establishing an irreducible floor on robust generalization.
This unity transforms the $\diamP$ from an abstract robustness hyperparameter into a concrete, measurable, and fundamental limit of learning.

\vspace{-3pt}
\section{Optimization}
\label{sec:optimization}
A central motivation for modeling distributional uncertainty is to optimize predictors against worst-case risk. Classical DRO formulates this objective as a minimax problem:
\begin{equation} \label{eq:dro-general}
\footnotesize
    h^{\mathrm{DRO}}
    \;:=\;
    \argmin_{h \in \mathcal{H}}
    \max_{Q \in \mathcal{U}_\rho(\hat{P})}
    L_Q(h),
\end{equation}

where $\mathcal{U}_\rho(\hat{P})$ is typically a Wasserstein or $f$-divergence ball of radius $\rho$ around the empirical distribution. While theoretically robust, this continuous formulation suffers from the \textit{specification problem}—tuning the radius $\rho$ is non-trivial, and worst-case perturbations often drift off the data manifold into non-semantic noise and the \emph{computational problem} of optimizing over an infinite-dimensional space of probability measures, requiring nontrivial duality arguments or approximations ~\citep{duchi2021learning, kuhn2019wasserstein, kuhn2025distributionally}. 

By defining the uncertainty set via the convex hull of coherent generative processes rather than an unstructured geometric ball, the intractable supremum collapses to a discrete, tractable, and interpretable min--max game.
\vspace{-3pt}
\subsection{Reduction to Extreme Points}
Leveraging the geometric structure of $\mathcal{P}$ established in \cref{def:structured-credal-set}, we derive an exact reduction of the optimization objective.

\begin{proposition}[DRO Reduction]\label{prp:dro-reduction}
Let $\mathcal{P} = \mathrm{Conv}(\{P_{ij}\}_{i \in [N_X], j \in [N_Y]})$ be the structured credal set. The infinite-dimensional DRO problem in \cref{eq:dro-general} reduces to a discrete minimax optimization over the extreme points:
\begin{equation}\label{eq:dro-minmax}
    h^{\mathrm{DRO}} = \argmin_{h \in \mathcal{H}} \max_{(i,j) \in [N_X] \times [N_Y]} L_{P_{ij}}(h).
\end{equation}
\end{proposition}
\color{black}
\vspace{-5pt} \noindent \textbf{Significance.} \cref{prp:dro-reduction} reinterprets the optimization as Structured Group DRO \cite{sagawa2019distributionally}, where groups are defined by generative factors (environment $i$, labeler $j$). 
Crucially, this yields a \textbf{hyperparameter-free} framework: the robustness radius is intrinsic—determined by the $ \diamP \le \eta_X + \eta_{Y|X}^{\mathrm{eff}}$—rather than a tuned hyperparameter. The inner maximization admits a clear adversarial interpretation—Nature selects the worst-case combination of environment and labeler—while remaining constrained to semantically plausible generative worlds, avoiding the physically impossible perturbations possible in Wasserstein DRO.

 \vspace{-5pt}

In the pure labeling-uncertainty regime, the DRO objective in \cref{prp:dro-reduction} reduces to enforcing robustness to the most challenging labeler under the current hypothesis, discouraging overfitting to majority opinion while guaranteeing performance on minority labelers—a property particularly important in medical imaging, where rare but critical findings may be identified by only a subset of expert radiologists and can be washed out by majority voting.

\paragraph{Algorithmic strategies.} Because the inner maximization in Equation~\eqref{eq:dro-minmax} is taken over a finite set of extreme points, the resulting minimax problem admits straightforward finite-world optimization via adversarial reweighting over structured worlds. In particular, the learner may be viewed as minimizing a weighted mixture of per-world risks, $\sum_{i,j} w_{ij}\,\hat L_{ij}(h)$, where $w$ is the adversary's strategy over worlds $(i,j)$. A natural hard-minimax strategy is greedy worst-world descent: at each iteration, identify the current worst world $(i^\ast,j^\ast)$ and update using only its gradient (equivalently, a subgradient step on the pointwise maximum), which is computationally cheap but can yield non-smooth dynamics. To stabilize training, one can instead optimize the Log-Sum-Exp surrogate
\begin{equation}\label{eq:lse_surrogate}
    \mathcal{L}_\tau(h)
    =
    \tau \log \sum_{(i,j)}
    \exp\!\left( \frac{\hat{L}_{ij}(h)}{\tau} \right).
\end{equation}
whose gradient is a weighted average of per-world gradients. As $\tau \!\downarrow\! 0$ this recovers hard minimax training; larger $\tau$ interpolates toward uniform averaging across worlds.

\color{black}

\vspace{-3pt}
\section{Synthetic Experiments}
\vspace{-3pt}

We also empirically validated the central geometric claims of the structured credal framework and assessed the statistical behavior of the proposed estimators on controlled synthetic families. These experiments are designed to verify the theoretical bounds in Sections \ref{sec:theory} and \ref{sec:nx1} under known ground truth, where all quantities can be computed analytically. 
\vspace{-4pt}
\noindent \paragraph{~\cref{prp:general-bounds} (General Distance Bounds.)} First, we characterize the empirical tightness of the general distance decomposition by constructing a structured credal set using $20$ Gaussian environments and $10$ labelers, including soft (sigmoid) and deterministic hard (threshold) labelers with parameters sampled on uniform grids resulting in $40{,}000$ joint distribution pairs. Across all evaluated cases, the empirical total variation distances satisfied the theoretical bounds with essentially zero slack (up to numerical precision) in the pure environment-shift and pure labeler-shift regimes and only moderate slack in the simultaneous joint-shift setting. Summarized results are shown in Table \ref{tab:prop-43-tightness} with detailed breakdown in \cref{tab:gen-dist-bounds} in the Appendix. 

\begin{table}[h]
\centering
\caption{Tightness of general distance bounds (\cref{prp:general-bounds}) in the joint-shift regime ($i \neq i',\, j \neq j'$). $\Delta_{\mathrm{low}}$ and $\Delta_{\mathrm{up}}$ denote the mean gap between the true TV distance and the lower/upper bound, respectively.}
\label{tab:prop-43-tightness}
\begin{tabular}{lcc}
\toprule
\textbf{Labeler Regime} & \textbf{$\Delta_{\mathrm{low}} \downarrow$} & \textbf{$\Delta_{\mathrm{up}} \downarrow$} \\ \midrule
Soft             & 0.242  & 0.165  \\
Det.\ Hard       & 0.228  & 0.096 \\ \bottomrule
\end{tabular}
\end{table}

\color{black}

\vspace{-4pt}
\noindent \paragraph{Disagreement-diameter characterization.} We show that in the fixed-covariate regime the empirical diameter estimator is nearly unbiased for deterministic hard labelers (Table~\ref{tab:det_diameter-ablation}), remains accurate in the soft-label setting via the \(L_1\)-based form (Table~\ref{tab:softlabel-kappa-ablation}), and continues to track the closed-form noisy-annotator prediction with low error across increasing noise budgets (Table~\ref{tab:noise-ablation}). 

\vspace{-4pt}
\noindent \paragraph{\cref{thm:concentration} (Finite Sample Concentration \& Labeler complexity).} Finally, the absolute estimation error decays at the predicted \(O(n^{-1/2})\) rate across multiple covariate distributions (Tables \ref{tab:det_n01_samp_complexity}, \ref{tab:det_n02_samp_complexity}, \ref{tab:det_n21_samp_complexity}, \ref{tab:iv21} and \ref{tab:prob_samp_comp}; as well as Figure~\ref{fig:app:concentration}), with negligible empirical violation probability and stable tightness ratios, and we additionally study scaling with the number of labeling mechanisms, observing non-vacuous bounds and robustness even for large \(N_Y\) (Appendix~\ref{sec:app:results}).





\color{black}

\vspace{-3pt}
\section{Conclusion}
\vspace{-3pt}
We introduced a structured credal learning framework for decomposing the distributional uncertainty into covariate distributions and labeling mechanisms: By leveraging structure, we derived sharp geometric bounds on the diameter of the induced credal set and established an exact characterization in the pure label uncertainty regime. This result transforms robustness from a hyperparameter-dependent quantity into a directly measurable statistic governed by annotator disagreement. We further show that robust optimization over structured credal sets admits a tractable discrete formulation, replacing ad-hoc tuned robustness-radius with an intrinsic robustness quantity determined by the structured credal set yielding interpretable worst-case objectives. Together, these results provide a principled theoretical foundation for robust learning under simultaneous covariate shift and label uncertainty, bridging the gap between abstract robustness guarantees and empirically observable uncertainty sources. 
\vspace{-3pt}
\section{Limitations and Future Work} 
\vspace{-3pt}
Our framework currently assumes finite collections of environments and labelers, which may limit scalability to settings with continuous or high-dimensional uncertainty families. While the structured credal construction improves interpretability, it relies on explicit specification or estimation of candidate environments and labelers, which may be challenging in fully unsupervised or highly dynamic deployment scenarios. 

Several directions naturally follow from this work. First, extending the structured credal framework to continuous families of environments and annotators, potentially parameterized by latent variables, would broaden its applicability to large-scale deployment settings. Second, integrating the proposed robustness formulation with modern representation learning pipelines (e.g., self-supervised models) may enable joint optimization of feature invariance and label disagreement robustness. Another promising direction is the incorporation of task-specific structure, such as hierarchical labels or structured outputs, into the labeling mechanism component of the credal set. This may yield finer-grained robustness guarantees in complex prediction tasks.


\bibliography{icml2026}
\bibliographystyle{icml2026}

\newpage
\appendix
\onecolumn

\section{Proofs from Sections \ref{sec:framework} and \ref{sec:theory}}

\subsection{Lemma \ref{lem:extremum}}
\begin{proof}
Since $\mathcal{P}$ is the convex hull of finitely many probability measures $\{P_{ij}\}_{(i,j) \in [N_X] \times [N_Y]}$, it is a compact convex polytope and its finite set of extreme points satisfies $\mathrm{ex}(\mathcal{P}) \subseteq \{P_{ij}\}$.

\medskip\noindent\textbf{Case (i)}
The function $f: \mathcal{P} \to \mathbb{R}$ is continuous and convex. By the Bauer Maximum Principle,
\(
  \sup_{Q \in \mathcal{P}} f(Q) = \max_{Q \in \mathrm{ex}(\mathcal{P})} f(Q).
\)

\medskip\noindent\textbf{Case (ii)}
Let $f\colon \mathcal{P}^2 \to \mathbb{R}$ be continuous and separately convex. Fix the second argument $Q_2 \in \mathcal{P}$. The partial map $Q_1 \mapsto f(Q_1, Q_2)$ is continuous and convex on $\mathcal{P}$. Applying the $k=1$ case, it attains its supremum at an extreme point $Q_1^\star \in \mathrm{ex}(\mathcal{P})$. We define this partial supremum as:
\[
  g(Q_2) = \max_{Q_1 \in \mathrm{ex}(\mathcal{P})} f(Q_1, Q_2).
\]
Because $\mathrm{ex}(\mathcal{P})$ is a finite set, $g(Q_2)$ is the pointwise maximum of a finite number of functions, each of which is continuous and convex in $Q_2$. Since the pointwise maximum of finitely many continuous, convex functions is itself continuous and convex, $g\colon \mathcal{P} \to \mathbb{R}$ satisfies the conditions of the $k=1$ case. 

Applying the Bauer Maximum Principle to $g$ over $Q_2 \in \mathcal{P}$, the supremum is attained at an extreme point $Q_2^\star \in \mathrm{ex}(\mathcal{P})$, yielding:
\[
  \sup_{(Q_1, Q_2) \in \mathcal{P}^2} f(Q_1, Q_2) = \sup_{Q_2 \in \mathcal{P}} g(Q_2) = \max_{Q_2 \in \mathrm{ex}(\mathcal{P})} g(Q_2) = \max_{(Q_1, Q_2) \in \mathrm{ex}(\mathcal{P})^2} f(Q_1, Q_2).
\]
This completes the proof.
\end{proof}

\subsection{Proofs for Propositions \& Theorems \ref{prp:label-disagreement} - \ref{prp:diameter-decomposition}} \label{proofs:general}
\subsubsection{Proof for Proposition \ref{prp:label-disagreement}}
\label{proof:label-disagreement}

Using the density representation with respect to product measure $ \mu_\Omega =  \mu_X \otimes \mu_Y$:
\begin{align*}
\dTV(P_{ij}, P_{ij'}) &= \frac{1}{2} \int_\calX \int_\calY \left|p_X^{(i)}(x) p_{Y|X}^{(j)}(y|x) - p_X^{(i)}(x) p_{Y|X}^{(j')}(y|x)\right| d\mu_Y(y) \, d\mu_X(x).
\end{align*}

Since $p_X^{(i)}(x) \geq 0 \,\forall  \,x \in \mathcal{X}$ as it is a probability density, we can factor it out of the absolute value:
\begin{align*}
&= \frac{1}{2} \int_\calX p_X^{(i)}(x) \Bigg(\int_\calY \left|p_{Y|X}^{(j)}(y|x) - p_{Y|X}^{(j')}(y|x)\right| d\mu_Y(y) \Bigg) \, d\mu_X(x).
\end{align*}

Recognizing that $\frac{1}{2}\int_\calY |p_{Y|X}^{(j)}(y|x) - p_{Y|X}^{(j')}(y|x)| d\mu_Y(y) = \dTV(P_{Y|X}^{(j)}(\cdot|x), P_{Y|X}^{(j')}(\cdot|x))$:
\begin{align*}
&= \int_\calX p_X^{(i)}(x) \cdot \dTV\left(P_{Y|X}^{(j)}(\cdot|x), P_{Y|X}^{(j')}(\cdot|x)\right) d\mu_X(x) 
= \mathbb{E}_{X \sim P_X^{(i)}} \left[\dTV\left(P_{Y|X}^{(j)}(\cdot|X), P_{Y|X}^{(j')}(\cdot|X)\right)\right].
\end{align*}

\subsubsection{Proof for Proposition~\ref{prp:covariate-shift} }
\label{proof:covariate-shift}

Using the density representation with respect to product measure  $\mu_\Omega$, the joint densities factorize as $p_{ij}(x,y) = p_X^{(i)}(x) p_{Y|X}^{(j)}(y|x)$ and $p_{i'j}(x,y) = p_X^{(i')}(x) p_{Y|X}^{(j)}(y|x)$.

Plugging this into the definition of TV distance:
\begin{align*}
\dTV(P_{ij}, P_{i'j}) &= \frac{1}{2} \int_\calX \int_\calY \left| p_X^{(i)}(x) p_{Y|X}^{(j)}(y|x) - p_X^{(i')}(x) p_{Y|X}^{(j)}(y|x) \right| \, d\mu_Y(y) \, d\mu_X(x).
\end{align*}

Since the conditional density $p_{Y|X}^{(j)}(y|x)$ is non-negative, we factor it out:
\begin{align*}
&= \frac{1}{2} \int_\calX \left| p_X^{(i)}(x) - p_X^{(i')}(x) \right| \left( \int_\calY p_{Y|X}^{(j)}(y|x) \, d\mu_Y(y) \right) \, d\mu_X(x).
\end{align*}

As $\int_\calY p_{Y|X}^{(j)}(y|x) \, d\mu_Y(y) = 1$, substituting 1 in the above equation yields:

\begin{align}
\label{eq:cov_ineqality}
&=\frac{1}{2} \int_\calX \left| p_X^{(i)}(x) - p_X^{(i')}(x) \right| \cdot 1 \, d\mu_X(x) \\
&= \dTV\left(P_X^{(i)}, P_X^{(i')}\right).
\end{align}


\subsubsection{Proof for Theorem ~\ref{prp:general-bounds}}

\begin{proof}
We proceed through two decomposition paths and combine the resulting bounds.

\textbf{Lower bound.}

\medskip
\noindent\textbf{Step 1: Establishing exact intermediate distances.}

\noindent By Proposition~\ref{prp:label-disagreement} applied with shared
$P_X^{(i)}$ and $P_X^{(i')}$ respectively, let 
\begin{equation}
    d_{\mathrm{TV}}(P_{ij},\, P_{ij'}) = A_i,
    \qquad
    d_{\mathrm{TV}}(P_{i'j},\, P_{i'j'}) = A_{i'}.
    \label{eq:label-exact}
\end{equation}
By Proposition~\ref{prp:covariate-shift}
(exact covariate-shift distance with shared labeler), applied with shared
$P_{Y|X}^{(j)}$ and $P_{Y|X}^{(j')}$ respectively, let
\begin{equation}
    d_{\mathrm{TV}}(P_{ij},\, P_{i'j}) = C,
    \qquad
    d_{\mathrm{TV}}(P_{ij'},\, P_{i'j'}) = C.
    \label{eq:covariate-exact}
\end{equation}

\medskip
\noindent\textbf{Step 2: Path A (change labeler first, then environment).}

\noindent Consider the intermediate chain
$P_{ij} \to P_{ij'} \to P_{i'j'}$.
By the reverse triangle inequality for the total variation metric:
\begin{equation}
    d_{\mathrm{TV}}(P_{ij}, P_{i'j'})
    \;\geq\;
    \left|
      d_{\mathrm{TV}}(P_{ij}, P_{ij'})
      -
      d_{\mathrm{TV}}(P_{ij'}, P_{i'j'})
    \right|.
    \label{eq:reverse-tri-A}
\end{equation}
Substituting the exact values from~\eqref{eq:label-exact}
and~\eqref{eq:covariate-exact} into~\eqref{eq:reverse-tri-A}:
\begin{equation}
    d_{\mathrm{TV}}(P_{ij}, P_{i'j'})
    \;\geq\;
    \left| A_i - C \right|.
    \tag{Path A}
    \label{eq:path-A}
\end{equation}

\medskip
\noindent\textbf{Step 3: Path B (change environment first, then labeler).}

\noindent Consider the intermediate chain
$P_{ij} \to P_{i'j} \to P_{i'j'}$.
By the reverse triangle inequality:
\begin{equation}
    d_{\mathrm{TV}}(P_{ij}, P_{i'j'})
    \;\geq\;
    \left|
      d_{\mathrm{TV}}(P_{ij}, P_{i'j})
      -
      d_{\mathrm{TV}}(P_{i'j}, P_{i'j'})
    \right|.
    \label{eq:reverse-tri-B}
\end{equation}
Substituting the exact values from~\eqref{eq:label-exact}
and~\eqref{eq:covariate-exact} into~\eqref{eq:reverse-tri-B}:
\begin{equation}
    d_{\mathrm{TV}}(P_{ij}, P_{i'j'})
    \;\geq\;
    \left| C - A_{i'} \right|
    = \left| A_{i'} - C \right|.
    \tag{Path B}
    \label{eq:path-B}
\end{equation}

\medskip

Taking the maximum over both paths and substituting the definitions yields the stated bound.
\color{black}

\textbf{Upper Bound.} By the triangle inequality for total variation distance:
\[
\dTV(P_{ij}, P_{i'j'}) \leq \dTV(P_{ij}, P_{i'j}) + \dTV(P_{i'j}, P_{i'j'}).
\]

Applying \Cref{prp:covariate-shift} to the first term:
\[
\dTV(P_{ij}, P_{i'j}) = \dTV(P_X^{(i)}, P_X^{(i')}).
\]

Applying \Cref{prp:label-disagreement} to the second term:
\[
\dTV(P_{i'j}, P_{i'j'}) = \mathbb{E}_{X \sim P_X^{(i')}}\left[\dTV(P_{Y|X}^{(j)}, P_{Y|X}^{(j')})\right].
\]

Alternatively, using the path through $P_{ij'}$:
\[
\dTV(P_{ij}, P_{i'j'}) \leq \dTV(P_{ij}, P_{ij'}) + \dTV(P_{ij'}, P_{i'j'}),
\]

which gives:
\[
\dTV(P_{ij}, P_{i'j'}) \leq \mathbb{E}_{X \sim P_X^{(i)}}\left[\dTV(P_{Y|X}^{(j)}, P_{Y|X}^{(j')})\right] + \dTV(P_X^{(i)}, P_X^{(i')}).
\]

Taking the minimum over both paths and substituting the definitions yields the stated bound.
\end{proof}

\subsubsection{Proof for Theorem \ref{prp:diameter-decomposition}}
\label{proof:diameter-decomposition}

\begin{definition}[Pointwise Label Disagreement]
\label{def:pointwise-disagreement}
For each $x \in \calX$, the \emph{pointwise label disagreement} is:
\begin{equation}
\eta_{Y|X}(x) := \max_{j, j' \in [N_Y]} \dTV(P_{Y|X}^{(j)}, P_{Y|X}^{(j')}).
\end{equation}
\end{definition}

\begin{definition}[Maximum Expected Label Disagreement under $P_X^{(i)}$]
\label{def:expected-disagreement}
The \emph{max expected label disagreement under $P_X^{(i)}$} is:
\begin{equation}
\eta_{Y|X}^{(i)} := \max_{j, j' \in [N_Y]} \E_{X \sim P_X^{(i)}}\left[\dTV(P_{Y|X}^{(j)}, P_{Y|X}^{(j')})\right].
\end{equation}
\end{definition}

\textbf{Lower Bound 1: $\mathrm{diam}_{\mathrm{TV}}(\mathcal{P}) \geq \eta_X$.}

The diameter is achieved at extreme points. Consider any pair $(i, i')$ with $i \neq i'$ and fix any $j \in [N_Y]$:
\[
\mathrm{diam}_{\mathrm{TV}}(\mathcal{P}) \geq \dTV(P_{ij}, P_{i'j}) = \dTV(P_X^{(i)}, P_X^{(i')}),
\]
where the equality follows from \Cref{prp:covariate-shift}. Taking the maximum over all pairs $(i, i')$:
\[
\mathrm{diam}_{\mathrm{TV}}(\mathcal{P}) \geq \max_{i, i' \in [N_X]} \dTV(P_X^{(i)}, P_X^{(i')}) = \eta_X.
\]

\textbf{Lower Bound 2: $\mathrm{diam}_{\mathrm{TV}}(\mathcal{P}) \geq \eta_{Y|X}^*$}

Fix any $i \in [N_X]$ and consider pairs $(j, j')$ with $j \neq j'$:
\[
\mathrm{diam}_{\mathrm{TV}}(\mathcal{P}) \geq \dTV(P_{ij}, P_{ij'}) = \mathbb{E}_{X \sim P_X^{(i)}}\left[\dTV(P_{Y|X}^{(j)}, P_{Y|X}^{(j')})\right],
\]
where the equality follows from \cref{prp:label-disagreement}. Taking the maximum over $(j, j')$:

\[
\mathrm{diam}_{\mathrm{TV}}(\mathcal{P}) \geq \max_{j, j'} \mathbb{E}_{X \sim P_X^{(i)}}\left[\dTV(P_{Y|X}^{(j)}, P_{Y|X}^{(j')})\right] = \eta_{Y|X}^{(i)}.
\]
Since this holds for all $i$, taking the maximum over $i$:
\[
\mathrm{diam}_{\mathrm{TV}}(\mathcal{P}) \geq \max_{i \in [N_X]} \eta_{Y|X}^{(i)} = \eta_{Y|X}^*
\]

Combining both lower bounds:
\[
\mathrm{diam}_{\mathrm{TV}}(\mathcal{P}) \geq \max\{\eta_X, \eta_{Y|X}^*\}.
\]

\textbf{Upper Bound 1: $ \mathrm{diam}_{TV}(\mathcal{P}) \le \eta_X + \eta_{Y|X}^\ast$.}

Let $P_{ij}$ and $P_{i'j'}$ be arbitrary extreme points of $\mathcal P$. By the triangle inequality,
\[
d_{TV}(P_{ij},P_{i'j'})
\;\le\;
d_{TV}(P_{ij},P_{i'j}) + d_{TV}(P_{i'j},P_{i'j'}).
\]

For the first term, since the conditional distribution is shared, \cref{prp:covariate-shift} yields
\[
d_{TV}(P_{ij},P_{i'j})
\;\le\;
d_{TV}(P_X^{(i)},P_X^{(i')})
\;\le\;
\eta_X.
\]

For the second term, the marginal distribution is shared, and by \cref{prp:label-disagreement},
\[
d_{TV}(P_{i'j},P_{i'j'})
=
\mathbb E_{X\sim P_X^{(i')}}\!\left[
d_{TV}\!\left(
P_{Y|X}^{(j)}(\cdot|X),
P_{Y|X}^{(j')}(\cdot|X)
\right)
\right]
\;\le\;
\eta^{\ast}_{Y|X}.
\]

Combining the two bounds gives us:
\[
d_{TV}(P_{ij},P_{i'j'})
\;\le\;
\eta_X + \eta^{\ast}_{Y|X}.
\]
Since the supremum of total variation over a convex set is attained at extreme points, the result follows.

\textbf{Upper Bound 2: $\mathrm{diam}_{TV}(\mathcal{P}) \le \eta_X + (1-\eta_X) \bar{\eta}_{Y|X}$.}

Using the density representation with respect to product measure  $\mu_\Omega$,  We define the ``overlap'' density $m(x)$ and the ``excess'' densities $\Delta_i(x)$ and $\Delta_{i'}(x)$ as follows:
\[
m(x) := \min\{p_X^{(i)}(x), p^{(i')}_X(x)\}, 
\qquad
\Delta_i(x) := p_X^{(i)}(x) - m(x), 
\qquad
\Delta_{i'}(x) := p^{(i')}_X(x) - m(x)
\]
By construction $\Delta_i(x) \cdot \Delta_{i'}(x)=0$ for all $x$, and
\(
\int_{\calX} m(x)\, d\mu_X(x) = 1 - \dTV(P_X^{(i)}, P_X^{(i')}).
\)
We decompose the Total Variation distance using these densities:
\begin{align*}
\dTV(P_{ij}, P_{i'j'}) &= \frac{1}{2} \int_{\calX} \int_{\calY} |p_i(x)p_{Y|X}^{(j)}(y|x) - p_{i'}(x)p_{Y|X}^{(j')}(y|x)| \, d\mu_Y(y) \, d\mu_X(x) \\
&= \frac{1}{2} \int_{\calX} \int_{\calY} |(m(x) + \Delta_i(x))p_{Y|X}^{(j)} - (m(x) + \Delta_{i'}(x))p_{Y|X}^{(j')}| \, d\mu_Y(y) \, d\mu_X(x)
\end{align*}
Applying the triangle inequality:
\begin{align*}
\dTV(P_{ij}, P_{i'j'}) &\leq \underbrace{\frac{1}{2} \int_{\calX} m(x) \int_{\calY} |p_{Y|X}^{(j)} - p_{Y|X}^{(j')}| \, d\mu_Y(y) \, d\mu_X(x)}_{\text{Term I}} \\
&+ \underbrace{\frac{1}{2} \int_{\calX} \int_{\calY} |\Delta_i(x)p_{Y|X}^{(j)} - \Delta_{i'}(x)p_{Y|X}^{(j')}| \, d\mu_Y(y) \, d\mu_X(x)}_{\text{Term II}}
\end{align*}

\textbf{Analyzing Term I (Label Disagreement on Overlap):}
We regroup the factor of $1/2$ into the inner integral to identify the conditional Total Variation distance:
\begin{align*}
\text{Term I} &= \int_{\calX} m(x) \underbrace{\left( \frac{1}{2} \int_{\calY} |p_{Y|X}^{(j)}(y|x) - p_{Y|X}^{(j')}(y|x)| \, d\mu_Y(y) \right)}_{\dTV(P_{Y|X}^{(j)}(\cdot|x), P_{Y|X}^{(j')}(\cdot|x))} \, d\mu_X(x) \\
&= \int_{\calX} m(x) \dTV(P_{Y|X}^{(j)}(\cdot|x), P_{Y|X}^{(j')}(\cdot|x)) \, d\mu_X(x) \\
&\leq \int_{\calX} m(x) \max_{j, j'} \dTV(P^{(j)}_{Y|X}(\cdot|x),P^{(j')}_{Y|X}(\cdot|x)) \, d\mu_X(x) \\
  &= \int_{\calX} m(x) \eta_{Y|X}(x) \, d\mu_X(x)  \\
  & \leq \int_{\calX} m(x) \sup_{x \in \calX} \eta_{Y|X}(x) \, d\mu_X(x) \\
  &= \int_{\calX} m(x) \bar{\eta}_{Y|X} \, d\mu_X(x) \\
&= \bar{\eta}_{Y|X} (1 - \dTV(P_X^{(i)}, P_X^{(i')})) \\
\end{align*}

\begin{figure}[!h]
     \centering
     \begin{subfigure}{0.8\textwidth}
         \centering
         \includegraphics[width=\linewidth]{chapters/figures/gating_effect_1d.pdf}
         \caption{\textbf{Gating Effect under Covariate Shift.} Joint total variation distance varies with the position of the covariate window, peaking when probability mass overlaps regions of high labeling disagreement. While the covariate-only TV remains nearly constant, the joint TV closely follows the theoretical upper bound, illustrating how covariate shift modulates the observable impact of labeling uncertainty.}
         \label{fig:app:gating:1d}
     \end{subfigure}
     \hfill
     \begin{subfigure}{0.99\textwidth}
         \centering
         \includegraphics[width=\linewidth]{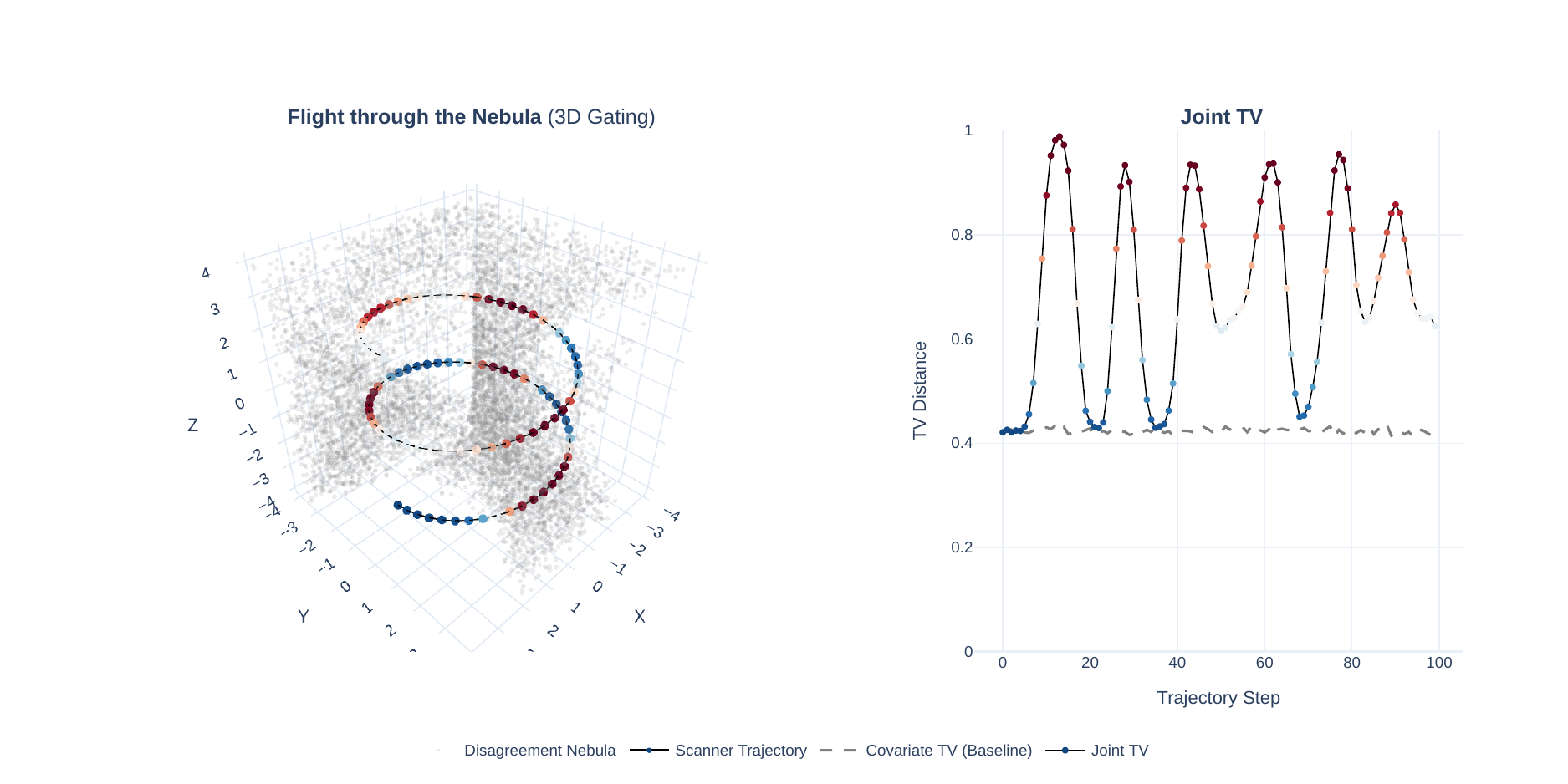}
         \caption{\textbf{3D Gating Visualization}: The left panel illustrates a scanner trajectory passing through a three-dimensional disagreement region, while the right panel shows the corresponding joint total variation (TV) distance along the trajectory. Peaks in joint TV occur when the trajectory intersects high-disagreement regions, whereas outside these regions the divergence remains close to the covariate-only baseline, demonstrating the gating effect of covariate mass on observable labeling uncertainty.}
         \label{fig:app:gating:3d}
     \end{subfigure}
     \caption{Proposition~\ref{prp:general-bounds} shows that joint distributional divergence results from the interaction between covariate shift and labeling disagreement, with bounds corresponding to two interpolation paths between feature and label changes. Label disagreement affects the joint distance only on regions with non-negligible covariate mass, yielding a \textbf{gating effect} in which covariate shift controls the visibility of labeling uncertainty. Consequently, training distributions may underestimate risk when deployment concentrates probability mass in high-disagreement regions, a coupling made explicit by the structured credal framework.}
     \label{fig:app:gating}
\end{figure}

\textbf{Analyzing Term II (Feature Mismatch):}

\begin{enumerate}
    \item \textbf{Step 1: Understanding the Disjoint Supports:}
    For any specific input $x$, only one of three cases can occur:
    \begin{itemize}
        \item \textbf{Case A:} $p^{(i)}_X(x) > p^{(i')}_X(x)$. Here, the minimum is $p^{(i')}_X(x)$. Therefore, $\Delta_{i'}(x) = 0$ and $\Delta_i(x) > 0$.
        \item \textbf{Case B:} $p^{(i')}_X(x) > p^{(i)}_X(x)$. Here, the minimum is $p^{(i)}_X(x)$. Therefore, $\Delta_i(x) = 0$ and $\Delta_{i'}(x) > 0$.
        \item \textbf{Case C:} $p^{(i')}_X(x) = p^{(i)}_X(x)$. Here,  $\Delta_i(x) = \Delta_{i'}(x) = 0$. 
    \end{itemize}
    Crucially, this means the product $\Delta_i(x) \cdot \Delta_{i'}(x)$ is always zero. They never exist simultaneously at the same point $x$.

    \item \textbf{Step 2: Why the Absolute Value Vanishes}
    Consider the expression inside the absolute value brackets:
    \[
    \left| \Delta_i(x) \cdot p_{Y|X}^{(j)}(y|x) - \Delta_{i'}(x) \cdot p_{Y|X}^{(j')}(y|x) \right|
    \]
    Using the property established above:
    \begin{itemize}
        \item In \textbf{Case A} (where $\Delta_{i'}=0$), the term becomes $|\Delta_i \cdot p^{(j)}_{Y|X}(y|x) - 0| = \Delta_i \, p^{(j)}_{Y|X}(y|x)$.
        \item In \textbf{Case B} (where $\Delta_i=0$), the term becomes $|0 - \Delta_{i'} p^{(j')}| = \Delta_{i'} \cdot \, p^{(j')}_{Y|X}(y|x)$.
        \item In \textbf{Case C}, (where $\Delta_i=\Delta_{i'}=0$), the term becomes $|0-0| = 0$
    \end{itemize}
    In all three cases, the result is mathematically identical to the \emph{sum} of the two terms (since at least one term is always zero, adding it changes nothing). Thus, we can drop the absolute value and simply add them:
    \[
    \left| \Delta_i(x) \cdot p^{(j)}_{Y|X} (x) - \Delta_{i'}(x) \cdot p^{(j')}_{Y|X} (x) \right| = \Delta_i(x) \cdot p_{Y|X}^{(j)}(y|x) + \Delta_{i'}(x) \cdot p_{Y|X}^{(j')}(y|x)
    \]

    \item \textbf{Step 3: Splitting the Integral:}
    We substitute this sum back into the main integral:
    \[
    \text{Term II} = \frac{1}{2} \int_{\calX} \int_{\calY} \left( \Delta_i(x)p_{Y|X}^{(j)}(y|x) + \Delta_{i'}(x)p_{Y|X}^{(j')}(y|x) \right) \, d\mu_Y(x) \, d\mu_X(x)
    \]
    By the linearity of integration, we can separate this into two distinct integrals:
    \[
    = \frac{1}{2} \left[ \int_{\calX} \Delta_i(x) \left(\int_{\calY} p_{Y|X}^{(j)}(y|x) d\mu_Y(y)\right) d\mu_X(x) + \int_{\calX} \Delta_{i'}(x) \left(\int_{\calY} p_{Y|X}^{(j')}(y|x) d\mu_Y(y)\right) d\mu_X(x) \right]
    \]

    \item \textbf{Integrating out $y$:}
    Since $p_{Y|X}$ is a probability distribution:
    \[
    \int_{\calY} p_{Y|X}^{(j)}(y|x) \, d\mu_Y(x) \leq 1 \quad \text{and} \quad \int_{\calY} p_{Y|X}^{(j')}(y|x) \, d\mu_Y(x) \leq 1
    \]
    Substituting these 1s into the equation simplifies it to integrals over just $x$:
    \[
    \leq \frac{1}{2} \left[ \int_{\calX} \Delta_i(x) \cdot 1 \, d\mu_X(x) + \int_{\calX} \Delta_{i'}(x) \cdot 1 \, d\mu_X(x) \right]
    \]

    \item \textbf{Final Evaluation:}
    The integral of the excess density $\int \Delta_i(x) d\mu_X(x)$ represents the total probability mass unique to distribution $P_i$. A standard identity of Total Variation distance states that $\int \Delta_i(x) d\mu_X(x) = \dTV(P_X^{(i)}, P_X^{(i')})$.
    \[
    = \frac{1}{2} \left[ \dTV(P_X^{(i)}, P_X^{(i')}) + \dTV(P_X^{(i)}, P_X^{(i')}) \right]
    \]
    \[
    = \dTV(P_X^{(i)}, P_X^{(i')})
    \]
\end{enumerate}

Combining terms, for any pair of distributions:
\[
\dTV(P_{ij}, P_{i'j'}) \leq \dTV(P_X^{(i)}, P_X^{(i')}) + (1 - \dTV(P_X^{(i)}, P_X^{(i')}))\bar{\eta}_{Y|X}.
\]
Let $\delta = \dTV(P_X^{(i)}, P_X^{(i')})$. The function $f(\delta) = \delta + (1-\delta)\bar{\eta}_{Y|X}$ is monotonically increasing in $\delta$ (since $\bar{\eta}_{Y|X} \leq 1$). Thus, we maximize the bound by taking the maximum feature distance $\eta_X$:
\[
\mathrm{diam}_{\mathrm{TV}}(\mathcal{P}) \leq \eta_X + (1 - \eta_X)\bar{\eta}_{Y|X}.
\]

\textbf{Combined Upper Bound: $\mathrm{diam}_{TV}(\mathcal{P}) \le \eta_X + \min\{\eta_{Y|X}^\ast, (1 - \eta_X)\bar{\eta}_{Y|X} \}$}

Because $\mathrm{diam}_{\mathrm{TV}}(\mathcal{P})$ admits both quantities as upper bounds, taking their minimum and applying the identity $\min(a+b,a+c)=a+\min(b,c)$ yields
\[
\eta_X + \min\!\left\{\eta_{Y|X}^\ast,\; (1-\eta_X)\bar{\eta}_{Y|X}\right\},
\]
which establishes the claimed bound.

\section{Proofs for $N_X=1$ Regime}
\label{proof:pure_label_main}

We consider the \emph{pure label uncertainty} regime, where $N_X = 1$. In this case, the covariate distribution is fixed and unique, denoted by $P_X^{(1)} = P_X$, and therefore the covariate diameter satisfies $\eta_X = 0$.

\subsection{Proposition \ref{prop:geometry}: Geometric Exactness} \label{proof:geometric_exactness}

By \cref{prp:label-disagreement}, for any two joint distributions $P_{1j}$ and $P_{1j'}$ in the structured credal set that share the same marginal feature distribution $P_X$ but differ in their conditional label distributions, the total variation distance admits the exact representation
\[
d_{\mathrm{TV}}(P_{1j}, P_{1j'}) 
= 
\mathbb{E}_{X \sim P_X}
\!\left[
d_{\mathrm{TV}}\!\left(
P_{Y|X}^{(j)}(\cdot \mid X),
P_{Y|X}^{(j')}(\cdot \mid X)
\right)
\right].
\]

The diameter of the credal set $\mathcal{P}$ is defined as the maximum total variation distance between any two of its elements. Since $\mathcal{P}$ is the convex hull of the extreme points $\{P_{1j}\}_{j=1}^{N_Y}$ and total variation distance is maximized at extreme points, we obtain
\[
\mathrm{diam}_{\mathrm{TV}}(\mathcal{P})
=
\max_{j,j'}
d_{\mathrm{TV}}(P_{1j}, P_{1j'}).
\]

Substituting the expression above yields
\[
\mathrm{diam}_{\mathrm{TV}}(\mathcal{P})
=
\max_{j,j'}
\mathbb{E}_{X \sim P_X}
\!\left[
d_{\mathrm{TV}}\!\left(
P_{Y|X}^{(j)}(\cdot \mid X),
P_{Y|X}^{(j')}(\cdot \mid X)
\right)
\right]
=
\eta_{Y|X}^*,
\]
which completes the proof.

\subsection{Theorem \ref{thm:concentration}: Finite Sample Concentration} \label{proof:concentration}

To derive concentration bounds for the diameter of the structured credal set as defined in \cref{def:structured-credal-set} in the fixed-feature regime ($N_X=1$), we must first formalize the data-generating process involving multiple labeling mechanisms. Let $\mathcal{X}$ be the feature space and $\mathcal{Y}$ be the label space. We consider a single feature distribution $P_X$ and a set of $N_Y$ distinct labeling mechanisms $\{P_{Y|X}^{(k)}\}_{k=1}^{N_Y}$.

We posit the following assumptions regarding the sampling and annotation process:

\begin{assumption}[i.i.d. Feature Sampling]
\label{ass:iid_features}
The feature vectors are drawn independently and identically distributed from the marginal feature distribution:
$$
X_1, \dots, X_n \sim_{i.i.d.} P_X
$$
\end{assumption}

\begin{assumption}[Labeling Mechanism]
\label{ass:generative_labeling}
Each annotator $k \in \{1, \dots, K\}$ is characterized by a conditional distribution $P_{Y|X}^{(k)}$ over the label space $\mathcal{Y} = \{1, \dots, C\}$. For each sample $i \in \{1, \dots, n\}$ and input realization $X_i = x_i$:
\begin{enumerate}[label=(\alph*)]
    \item The annotator's \emph{belief} or \emph{confidence scores} is given by the probability vector:
    $$
    \mathbf{p}_k(x_i) := \left(P_{Y|X}^{(k)}(1 \mid x_i), \dots, P_{Y|X}^{(k)}(C \mid x_i)\right) \in \Delta^{C-1}
    $$
    where $\Delta^{C-1}$ is the $(C-1)$-dimensional probability simplex.
    
    \item The \emph{realized label} $y_i^{(k)} \in \mathcal{Y}$ is drawn according to:
    $$
    y_i^{(k)} \mid X_i = x_i \sim P_{Y|X}^{(k)}(\cdot \mid x_i)
    $$
\end{enumerate}
\end{assumption}

\begin{assumption}[Observation Model]
\label{ass:observation_model}
We observe one of the following, depending on the regime:
\begin{enumerate}[label=(\alph*)]
    \item \textbf{Soft Label Regime:} The full probability vector $\mathbf{p}_k(x_i) \in \Delta^{C-1}$ is observed for each sample $i$ and labeling mechanism $k$.
    
    \item \textbf{Hard Label Regime:} Only the realized label $y_i^{(k)} \in \mathcal{Y}$ is observed for each sample $i$ and labeling mechanism $k$.
\end{enumerate}
\end{assumption}

\begin{assumption}[Structure of Labeling Mechanism]
\label{ass:labeling_structure}
The conditional distributions $\{P_{Y|X}^{(k)}\}_{k=1}^K$ satisfy one of the following structural conditions:

\begin{enumerate}[label=(\alph*)]
    \item \textbf{Deterministic:} For each annotator $k$, there exists a measurable function $f_k: \mathcal{X} \to \mathcal{Y}$ such that:
    $$
    P_{Y|X}^{(k)}(c \mid x) = \mathbb{I}[f_k(x) = c] \quad \forall c \in \mathcal{Y}, \; x \in \mathcal{X}
    $$
    Equivalently, the conditional probability vector
    \[
    \mathbf{p}_k(x) := \big(P_{Y|X}^{(k)}(1 \mid x), \dots, P_{Y|X}^{(k)}(C \mid x)\big)
    \]
    belongs to $\{e_1, \dots, e_C\}$, where $e_c \in \mathbb{R}^C$ denotes the $c$-th standard basis vector (one-hot encoding).

    \item \textbf{Stochastic:} The conditional distributions are non-degenerate, i.e., there exists at least one $k \in [N_Y]$ and $x \in \mathcal{X}$ such that:
    $$
    P_{Y|X}^{(k)}(c \mid x) \in (0, 1) \quad \text{for some } c \in \mathcal{Y}
    $$
\end{enumerate}
\end{assumption}

\begin{assumption}[Conditional Independence of Labelers]
\label{ass:labeler_independence}
Conditional on the feature realization $X_i$, the labelers operate independently. That is, for any $j \neq k$:
$$
y_i^{(j)} \perp \!\!\! \perp y_i^{(k)} \mid X_i
$$
This implies that the joint conditional distribution of all $N_Y$ labels factorizes:
$$
P(y_i^{(1)}, \dots, y_i^{(K)} \mid X_i) = \prod_{k=1}^{N_Y} P_{Y|X}^{(k)}(y_i^{(k)} \mid X_i)
$$
\end{assumption}

\begin{assumption}[Sample Independence]
\label{ass:sample_independence}
The generation processes for distinct samples $i \neq j$ are mutually independent. Specifically, the tuple $\mathbf{Z}_i = (X_i, y_i^{(1)}, \dots, y_i^{(N_Y)})$ is independent of $\mathbf{Z}_j$.
\end{assumption}

Establishing that the fully observed tuples are i.i.d. allows us to apply standard concentration inequalities (e.g., Hoeffding's Inequality) to the estimation of pairwise disagreement rates.

\begin{lemma}[i.i.d. Structure of Annotated Tuples]
\label{lemma:iid_tuples}
Under Assumptions \ref{ass:iid_features}--\ref{ass:sample_independence}, the joint tuples $\mathbf{Z}_i = (X_i, y_i^{(1)}, \dots, y_i^{(K)})$ for $i=1, \dots, n$ are independent and identically distributed random variables taking values in the product space $\Omega_{joint} = \mathcal{X} \times \mathcal{Y}^K$.
\end{lemma}

\begin{proof}
Let $\mu$ be the base measure on $\mathcal{X}$ and $\nu$ be the base measure on $\mathcal{Y}$. We analyze the joint probability density (or mass) function of a single tuple $\mathbf{Z}_i$, denoted as $p_{\mathbf{Z}}(x, y^{(1)}, \dots, y^{(K)})$.

By the definition of conditional probability and Assumption \ref{ass:generative_labeling}, the joint density can be decomposed as:
$$
p_{\mathbf{Z}}(x, y^{(1)}, \dots, y^{(K)}) = p_X(x) \cdot p(y^{(1)}, \dots, y^{(K)} \mid x)
$$
Applying Assumption \ref{ass:labeler_independence} (Conditional Independence), the conditional term factorizes:
$$
p(y^{(1)}, \dots, y^{(K)} \mid x) = \prod_{k=1}^K p_{Y|X}^{(k)}(y^{(k)} \mid x)
$$
Thus, the joint distribution for any index $i$ is uniquely determined by the fixed distributions $P_X$ and $\{P_{Y|X}^{(k)}\}_{k=1}^K$:
\begin{equation}
\label{eq:joint_density}
p_{\mathbf{Z}}(x, y^{(1)}, \dots, y^{(K)}) = p_X(x) \prod_{k=1}^K p_{Y|X}^{(k)}(y^{(k)} \mid x)
\end{equation}
Since $P_X$ and $\{P_{Y|X}^{(k)}\}$ do not depend on the index $i$, every tuple $\mathbf{Z}_i$ follows the exact same distribution defined in Eq. \eqref{eq:joint_density}. Furthermore, by Assumption \ref{ass:sample_independence}, the tuples are mutually independent.

Therefore, the sequence $\mathbf{Z}_1, \dots, \mathbf{Z}_n$ consists of independent and identically distributed draws from the joint distribution $P_{\mathbf{Z}}$.
\end{proof}

\begin{definition}[Labeling Regimes]
\label{def:labeling_regimes}
The combination of Assumptions~\ref{ass:observation_model} and \ref{ass:labeling_structure} defines three distinct regimes:

\begin{center}
\begin{tabular}{lcc}
\toprule
\textbf{Regime} & \textbf{Observation} & \textbf{Structure} \\
\midrule
Soft Labels & $\mathbf{p}_k(x_i) \in \Delta^{C-1}$ & - \\
Deterministic Hard Labels & $y_i^{(k)} \in \mathcal{Y}$ & Deterministic \\
Stochastic Hard Labels & $y_i^{(k)} \in \mathcal{Y}$ & Stochastic \\
\bottomrule
\end{tabular}
\end{center}
\end{definition}

In the soft-label regime, we observe the conditional distributions themselves. For a deterministic hard-label regime, the observed label can be mapped to a conditional probability vector (refer to \cref{ass:labeling_structure}). For a classification task with $C$ classes, for each input $x_i$, the labeling mechanism $k$ provides a probability vector $\mathbf{p}_k(x_i) \in \Delta^{C-1}$, where $\mathbf{p}_k(x_i) \equiv P_{Y|X}^{(k)}(\cdot \mid x_i)$.

\begin{definition}[Empirical TV Distance]
The Total Variation distance between two probability vectors $\mathbf{p}, \mathbf{q}$ is defined as $d_{\mathrm{TV}}(\mathbf{p}, \mathbf{q}) := \frac{1}{2} \| \mathbf{p} - \mathbf{q} \|_1$.
For annotators $j$ and $k$, the empirical TV distance based on $n$ samples is:
\begin{equation}
\hat{d}_{TV}^{jk} = \frac{1}{n} \sum_{i=1}^{n} d_{TV}(\mathbf{p}_j(x_i), \mathbf{p}_k(x_i)).
\end{equation}
\end{definition}

\begin{definition}[Population TV Distance]
The true population disagreement is defined as the expected TV distance over the feature distribution $P_X$:
\begin{equation}
d_{TV}^{\star,jk} = \mathbb{E}_{x \sim P_X}\left[ d_{TV}(\mathbf{p}_j(x), \mathbf{p}_k(x)) \right].
\end{equation}
\end{definition}

We define the maximum disagreement over all pairs $(j,k)$ as:
\[ \hat{\eta}_{Y|X} = \max_{j,k} \hat{d}_{TV}^{jk}, \quad \eta^{\star}_{Y|X} = \max_{j,k} d^{\star, jk}_{TV} \]

\begin{theorem}[Concentration Bound]\label{thm:soft-conc}
Let $N_Y$ be the number of labeling mechanisms providing either deterministic hard labels or soft labels. For any $\delta \in (0,1)$, with probability at least $1 - \delta$:
\begin{equation}
|\hat{\eta}_{Y|X} - \eta^{\star}_{Y|X}| \leq \sqrt{\frac{\ln(N_Y(N_Y-1)/\delta)}{2n}}.
\end{equation}
\end{theorem}

\begin{proof}
We establish this result through analysis of the concentration of the empirical mean of the pairwise TV distances.

\textbf{Step 1: Define the random variables.}

For a fixed pair of annotators $(j,k)$, we define the random variable corresponding to the disagreement on the $i$-th sample:
\[
Z_i := d_{TV}(\mathbf{p}_j(x_i), \mathbf{p}_k(x_i)) = \frac{1}{2} \| \mathbf{p}_j(x_i) - \mathbf{p}_k(x_i) \|_1.
\]

We analyze the properties of $Z_i$:
\begin{enumerate}[topsep=0pt]
    \item \textbf{Boundedness:} By the definition of Total Variation distance on probability space, $Z_i \in [0, 1]$ for any valid probability vectors.
    \item \textbf{i.i.d.:} The annotator policies $\mathbf{p}_j(\cdot)$ and $\mathbf{p}_k(\cdot)$ are fixed, deterministic functions given $x_i$. Since the inputs $\{x_i\}_{i=1}^n$ are drawn i.i.d. from $P_X$ (Under Assumptions \ref{ass:iid_features}--\ref{ass:sample_independence}), the transformations $Z_i = g(x_i)$ are also i.i.d.
    \item \textbf{Expectation:} By linearity of expectation and the definition of the population parameter:
    \[ \mathbb{E}[Z_i] = \mathbb{E}_{x \sim P_X}[d_{TV}(\mathbf{p}_j(x), \mathbf{p}_k(x))] = d_{TV}^{\star,jk}. \]
\end{enumerate}

\textbf{Step 2: Apply Hoeffding's inequality.}
For bounded i.i.d. random variables $Z_i \in [0,1]$ with mean $\mu = d_{TV}^{\star,jk}$, Hoeffding's inequality states:
\begin{equation}
\mathbb{P}\left( \left| \frac{1}{n}\sum_{i=1}^n Z_i - \mu \right| > \epsilon \right) \leq 2\exp(-2n\epsilon^2).
\end{equation}

Substituting our quantities:

\begin{equation}
\mathbb{P}\left( \left| \hat{d}_{TV}^{jk} - d_{TV}^{\star,jk} \right| > \epsilon \right) \leq 2\exp(-2n\epsilon^2).
\end{equation}

\textbf{Step 3: Apply union bound over all pairs.}

There are $M = \binom{N_Y}{2} = \frac{N_Y(N_Y-1)}{2}$ distinct pairs of annotators. We apply the union bound to control the probability that \textit{any} pair exceeds the error threshold $\epsilon$.

Let $E_{jk}$ be the event $\{ | \hat{d}_{TV}^{jk} - d_{TV}^{\star,jk} | > \epsilon \}$.
\begin{align}
\mathbb{P}\left( \bigcup_{(j,k)} E_{jk} \right) &\leq \sum_{(j,k)} \mathbb{P}(E_{jk}) \nonumber \\
&\leq M \cdot 2\exp(-2n\epsilon^2) \nonumber \\
&= \frac{N_Y(N_Y-1)}{2} \cdot 2\exp(-2n\epsilon^2) \nonumber \\
&= N_Y(N_Y-1)\exp(-2n\epsilon^2).
\end{align}

We set this total probability to $\delta$ and solve for $\epsilon$:
\begin{align}
N_Y(N_Y-1)\exp(-2n\epsilon^2) &= \delta \\
\exp(-2n\epsilon^2) &= \frac{\delta}{N_Y(N_Y-1)} \\
-2n\epsilon^2 &= \ln\left(\frac{\delta}{N_Y(N_Y-1)}\right)N_Y\\
\epsilon &= \sqrt{\frac{\ln(N_Y(N_Y-1)/\delta)}{2n}}.
\end{align}

\textbf{Step 4: Bound on the maxima.}

With probability at least $1 - \delta$, the bound $\epsilon$ holds for all pairs $(j,k)$ simultaneously.
Using the property that for any two sequences, $|\max_i a_i - \max_i b_i| \leq \max_i |a_i - b_i|$, we obtain:
\begin{align}
|\hat{\eta}_{Y|X} - \eta^{\star}_{Y|X}| &= \left| \max_{j,k} \hat{d}_{TV}^{jk} - \max_{j,k} d_{TV}^{\star,jk} \right| \nonumber \\
&\leq \max_{j,k} \left| \hat{d}_{TV}^{jk} - d_{TV}^{\star,jk} \right| \nonumber \\
&\leq \epsilon \nonumber \\
&= \sqrt{\frac{\ln(N_Y(N_Y-1)/\delta)}{2n}}.
\end{align}
\end{proof}

\subsection{Corollary \ref{corr:robust_generalization}: Robust Generalization} \label{proof:robust_generalization}

Credal Learning Theory provides a high-probability bound on the expected risk of the empirical risk minimizer $\hat{h}$ when evaluated under a possibly misspecified
distribution $Q \in \mathcal{P}$, distinct from the distribution $P \in \mathcal{P}$ from which the training data were drawn. In particular, CLT controls the excess risk incurred due to distributional ambiguity within the credal set.

Specifically, with probability at least $1 - \delta$, the following robustness bound holds:
\[
L_Q(\hat{h}) - L_P(h^\ast)
\;\le\;
\epsilon^\ast(\delta, n, \mathcal{H})
\;+\;
\mathrm{diam}_{\mathrm{TV}}(\mathcal{P}),
\]
where $\epsilon^\ast(\delta, n, \mathcal{H})$ denotes the statistical estimation error,
depending on the confidence level $\delta$, the sample size $n$, and the hypothesis class
$\mathcal{H}$, and $\mathrm{diam}_{\mathrm{TV}}(\mathcal{P})$ is the total variation
diameter of the credal set.

Within the structured credal framework, in the pure label uncertainty regime
($N_X = 1$), the diameter admits the exact characterization
\[
\mathrm{diam}_{\mathrm{TV}}(\mathcal{P}) = \eta_{Y \mid X}^\ast.
\]
Substituting this expression into the bound above yields the desired result.

\subsubsection{Theorem \ref{thm:minimax}: Minimax Optimality} \label{proof:minimax_optimality}

\begin{proof}
Fix an arbitrary sample size $n\in\mathbb{N}$ and an arbitrary learning algorithm
$\mathcal{A}$ mapping datasets $\mathcal{D}_n=((X_1,Y_1),\dots,(X_n,Y_n))$ to (possibly
randomized) predictors.

\vspace{0.5em}
\noindent\textbf{Step 1: Construct a $N_X=1$ credal set with diameter $\eta^*_{Y|X}=\eta$.}

Let $(\mathcal{X},\mathcal{A}_X,P_X)$ be any probability space and fix any $\eta\in(0,1)$.
Choose a measurable set $S\in\mathcal{A}_X$ such that $P_X(S)=\eta$.
Define two \emph{deterministic} labeling mechanisms (hard labels)
\[
f_1(x)\equiv 0,
\qquad
f_2(x):=\mathbf{1}\{x\in S\}.
\]
Equivalently, define two conditional label distributions
\[
P^{(1)}_{Y|X}(\cdot\mid x)=\delta_{0},
\qquad
P^{(2)}_{Y|X}(\cdot\mid x)=\delta_{f_2(x)}.
\]
Let $P_{11}$ and $P_{12}$ be the corresponding joint distributions (Definition \ref{def:joint-from-components}):
\[
P_{1j}(A\times B) = \int_A P^{(j)}_{Y|X}(B\mid x)\, dP_X(x),
\qquad j\in\{1,2\}.
\]
Finally, define the (structured) credal set
\[
\mathcal{P} := \mathrm{Conv}\big(\{P_{11},P_{12}\}\big).
\]
We compute its fixed-feature diameter $\eta^*_{Y|X}$.
For $x\notin S$, $P^{(1)}_{Y|X}(\cdot\mid x)=P^{(2)}_{Y|X}(\cdot\mid x)=\delta_0$ so the TV
distance is $0$; for $x\in S$, we have $\delta_0$ vs.\ $\delta_1$ so the TV distance is $1$.
Hence
\[
\eta^*_{Y|X}
=
\mathbb{E}_{X\sim P_X}\!\left[
d_{\mathrm{TV}}\!\big(P^{(1)}_{Y|X}(\cdot\mid X),P^{(2)}_{Y|X}(\cdot\mid X)\big)
\right]
=
\mathbb{E}_{X\sim P_X}\big[\mathbf{1}\{X\in S\}\big]
=
P_X(S)
=
\eta.
\]
Thus $\mathrm{diam}_{\mathrm{TV}}(\mathcal{P})=\eta^*_{Y|X}=\eta$ in this construction
(cf.\ Proposition \ref{prop:geometry}).

\vspace{0.5em}
\noindent\textbf{Step 2: Represent an arbitrary output of the learner and compute risks.}

Given a dataset $\mathcal{D}_n$, let $\widehat{h}_{\mathcal{D}_n}:=\mathcal{A}(\mathcal{D}_n)$.
Allow $\widehat{h}_{\mathcal{D}_n}$ to be randomized; define the induced prediction
probability
\[
\pi_{\mathcal{D}_n}(x)\;:=\;\mathbb{P}\big(\widehat{h}_{\mathcal{D}_n}(x)=1\big)\in[0,1],
\qquad
\mathbb{P}\big(\widehat{h}_{\mathcal{D}_n}(x)=0\big)=1-\pi_{\mathcal{D}_n}(x).
\]
Under $P_{11}$ the true label is $Y=f_1(X)=0$ almost surely, so the (conditional) pointwise
misclassification probability at $x$ equals $\pi_{\mathcal{D}_n}(x)$. Therefore
\begin{equation}
\label{eq:risk_p11}
L_{P_{11}}(\widehat{h}_{\mathcal{D}_n})
=
\mathbb{E}_{X\sim P_X}\big[\pi_{\mathcal{D}_n}(X)\big].
\end{equation}
Under $P_{12}$ the true label is $Y=f_2(X)=\mathbf{1}\{X\in S\}$ almost surely, hence
\[
\mathbb{P}\big(\widehat{h}_{\mathcal{D}_n}(X)\neq Y\mid X=x\big)
=
\begin{cases}
\pi_{\mathcal{D}_n}(x), & x\notin S \ \ (Y=0),\\
1-\pi_{\mathcal{D}_n}(x), & x\in S \ \ (Y=1).
\end{cases}
\]
Therefore
\begin{equation}
\label{eq:risk_p12}
L_{P_{12}}(\widehat{h}_{\mathcal{D}_n})
=
\mathbb{E}_{X\sim P_X}\Big[
\pi_{\mathcal{D}_n}(X)\mathbf{1}\{X\notin S\} + (1-\pi_{\mathcal{D}_n}(X))\mathbf{1}\{X\in S\}
\Big].
\end{equation}

\vspace{0.5em}
\noindent\textbf{Step 3: A two-point testing lower bound}

Add \eqref{eq:risk_p11} and \eqref{eq:risk_p12}. For each fixed dataset $\mathcal{D}_n$,
\begin{align*}
L_{P_{11}}(\widehat{h}_{\mathcal{D}_n}) + L_{P_{12}}(\widehat{h}_{\mathcal{D}_n})
&=
\mathbb{E}_{X\sim P_X}\!\Big[
\pi_{\mathcal{D}_n}(X)
+
\pi_{\mathcal{D}_n}(X)\mathbf{1}\{X\notin S\}
+
(1-\pi_{\mathcal{D}_n}(X))\mathbf{1}\{X\in S\}
\Big]\\
&=
\mathbb{E}_{X\sim P_X}\!\Big[
2\pi_{\mathcal{D}_n}(X)\mathbf{1}\{X\notin S\}
+
\big(\pi_{\mathcal{D}_n}(X) + 1-\pi_{\mathcal{D}_n}(X)\big)\mathbf{1}\{X\in S\}
\Big]\\
&=
\mathbb{E}_{X\sim P_X}\!\Big[
2\pi_{\mathcal{D}_n}(X)\mathbf{1}\{X\notin S\}
+
\mathbf{1}\{X\in S\}
\Big]\\
&\ge
\mathbb{E}_{X\sim P_X}\big[\mathbf{1}\{X\in S\}\big]
=
P_X(S)
=
\eta.
\end{align*}
Thus, for \emph{every} dataset $\mathcal{D}_n$,
\begin{equation}
\label{eq:sum_risks_ge_eta}
L_{P_{11}}(\widehat{h}_{\mathcal{D}_n}) + L_{P_{12}}(\widehat{h}_{\mathcal{D}_n})
\ \ge\ \eta.
\end{equation}

\vspace{0.5em}
\noindent\textbf{Step 4: Take expectations over training data and lower bound the worst-case test risk.}

Now choose the training distribution $P:=P_{11}\in\mathcal{P}$. Taking expectations of
\eqref{eq:sum_risks_ge_eta} over $\mathcal{D}_n\sim P_{11}^{\otimes n}$ yields
\[
\mathbb{E}_{\mathcal{D}_n\sim P_{11}^{\otimes n}}
\big[L_{P_{11}}(\widehat{h}_{\mathcal{D}_n})\big]
+
\mathbb{E}_{\mathcal{D}_n\sim P_{11}^{\otimes n}}
\big[L_{P_{12}}(\widehat{h}_{\mathcal{D}_n})\big]
\ \ge\ \eta.
\]
Hence, the larger of the two expectations is at least half:
\begin{equation}
\label{eq:max_exp_ge_half_eta}
\max\left\{
\mathbb{E}_{\mathcal{D}_n\sim P_{11}^{\otimes n}}[L_{P_{11}}(\widehat{h}_{\mathcal{D}_n})],
\ \mathbb{E}_{\mathcal{D}_n\sim P_{11}^{\otimes n}}[L_{P_{12}}(\widehat{h}_{\mathcal{D}_n})]
\right\}
\ \ge\ \frac{\eta}{2}.
\end{equation}
Since $\{P_{11},P_{12}\}\subseteq\mathcal{P}$, we have
\[
\sup_{Q\in\mathcal{P}}
\mathbb{E}_{\mathcal{D}_n\sim P_{11}^{\otimes n}}
\big[L_Q(\widehat{h}_{\mathcal{D}_n})\big]
\ \ge\
\max\left\{
\mathbb{E}_{\mathcal{D}_n\sim P_{11}^{\otimes n}}[L_{P_{11}}(\widehat{h}_{\mathcal{D}_n})],
\ \mathbb{E}_{\mathcal{D}_n\sim P_{11}^{\otimes n}}[L_{P_{12}}(\widehat{h}_{\mathcal{D}_n})]
\right\}
\ \ge\ \frac{\eta}{2},
\]
where the last inequality is \eqref{eq:max_exp_ge_half_eta}.

\vspace{0.5em}
\noindent\textbf{Step 5: Insert the baseline $L_P^*$ and conclude the minimax lower bound.}
Because $P_{11}$ is realizable by the deterministic classifier $h_1(x)\equiv 0$ (which we may
assume lies in $\mathcal{H}$), we have $L_{P_{11}}^*=0$. Therefore
\[
\sup_{Q\in\mathcal{P}}
\mathbb{E}_{\mathcal{D}_n\sim P_{11}^{\otimes n}}
\Big[L_Q(\widehat{h}_{\mathcal{D}_n}) - L_{P_{11}}^*\Big]
=
\sup_{Q\in\mathcal{P}}
\mathbb{E}_{\mathcal{D}_n\sim P_{11}^{\otimes n}}
\big[L_Q(\widehat{h}_{\mathcal{D}_n})\big]
\ \ge\ \frac{\eta}{2}.
\]
Since $P_{11}\in\mathcal{P}$, the supremum over $P\in\mathcal{P}$ can only increase the value:
\[
\sup_{P,Q\in\mathcal{P}}
\mathbb{E}_{\mathcal{D}_n\sim P^{\otimes n}}
\Big[L_Q(\mathcal{A}(\mathcal{D}_n)) - L_P^*\Big]
\ \ge\ \frac{\eta}{2}.
\]
Finally, the above bound holds for \emph{every} learning algorithm $\mathcal{A}$, hence
taking $\inf_{\mathcal{A}}$ preserves it:
\[
\inf_{\mathcal{A}}
\ \sup_{P,Q\in\mathcal{P}}
\ \mathbb{E}_{\mathcal{D}_n\sim P^{\otimes n}}
\Big[L_Q(\mathcal{A}(\mathcal{D}_n)) - L_P^*\Big]
\ \ge\ \frac{\eta}{2}.
\]
Recalling that $\eta=\eta^*_{Y|X}$ in this construction, we obtain exactly
\[
\inf_{\mathcal{A}}
\ \sup_{P,Q\in\mathcal{P}}
\ \mathbb{E}_{\mathcal{D}_n\sim P^{\otimes n}}
\Big[L_Q(\mathcal{A}(\mathcal{D}_n)) - L_P^*\Big]
\ \ge\ \frac{1}{2}\,\eta^*_{Y|X},
\]
which proves the theorem.
\end{proof}

\subsection{Noisy Annotator Model (Aleatoric Uncertainty)}\label{sec:app:proof:noisylabel}
We specialize the general structured credal learning framework to the classical noisy-annotator setting, where label uncertainty arises from irreducible annotation noise rather than model misspecification. This instantiation yields closed-form expressions for the induced conditional ambiguity radius, enabling direct interpretation of robustness guarantees in terms of annotator error rates.

\begin{itemize}
\item \textbf{Model (Binary Symmetric Noisy Annotators):}
Let $Y=\{0,1\}$ and assume a latent ground-truth label $y^*\in\{0,1\}$.
Annotator $j\in\{1,\dots,K\}$ outputs $y^{(j)}$ through a binary symmetric channel
with error rate $\varepsilon_j\in[0,0.5]$:
\[
\mathbb P\!\left(y^{(j)}=y^*\mid y^*\right)=1-\varepsilon_j,\qquad
\mathbb P\!\left(y^{(j)}\neq y^*\mid y^*\right)=\varepsilon_j.
\]
Assume conditional independence of annotators given $y^*$.

\item \textbf{Annotator--Annotator Disagreement (Closed Form):}
For any pair $(j,j')$, disagreement occurs iff exactly one annotator flips $y^*$.
Thus,
\[
\mathbb P\!\left(y^{(j)}\neq y^{(j')}\right)
=\mathbb P\!\left(y^{(j)}=y^*,\,y^{(j')}\neq y^*\right)
+\mathbb P\!\left(y^{(j)}\neq y^*,\,y^{(j')}=y^*\right)
=(1-\varepsilon_j)\varepsilon_{j'}+\varepsilon_j(1-\varepsilon_{j'})
=\varepsilon_j+\varepsilon_{j'}-2\varepsilon_j\varepsilon_{j'}.
\]

\item \textbf{Credal Diameter / Robustness Radius:}
In the fixed-covariate regime ($N_X=1$), the structured credal diameter equals the
maximum pairwise disagreement across annotators:
\[
\eta_{Y|X}^*
=\max_{j,j'} \mathbb P\!\left(y^{(j)}\neq y^{(j')}\right)
=\max_{j,j'}\big(\varepsilon_j+\varepsilon_{j'}-2\varepsilon_j\varepsilon_{j'}\big).
\]

\item \textbf{Non-vacuous Upper Bound under Weak Supervision:}
If all annotators satisfy $\varepsilon_j\le \varepsilon_{\max}$, then
\[
\eta_{Y|X}^*
\le \max_{j,j'}\big(\varepsilon_{\max}+\varepsilon_{\max}-2\varepsilon_{\max}^2\big)
=2\varepsilon_{\max}-2\varepsilon_{\max}^2.
\]
In particular, for $\varepsilon_{\max}\le 1/2$, this bound is strictly below $1$,
yielding a non-vacuous ambiguity radius even with noisy supervision.
\end{itemize}

\begin{proof}
The disagreement identity follows conditional independence given $y^*$ and a
two-term partition on which the annotator flips. The upper bound is immediate by
monotonicity in each $\varepsilon_j$ over $[0,0.5]$.
\end{proof}

\section{Proof for Proposition \ref{prp:dro-reduction}}
\begin{proof}
For any fixed hypothesis $h \in \mathcal{H}$, the expected risk $L_P(h) = \mathbb{E}_P[\ell((x,y), h)]$ is a strictly linear functional with respect to the probability measure $P$. With respect to the bounded loss-function ($\ell \in [0,1]$), the mapping $P \mapsto L_P(h)$ is continuous. 

Since $L_P(h)$ is continuous and linear (and thus convex) on the compact structured credal set $\mathcal{P}$, we can invoke \cref{lem:extremum} with $k=1$. Thus the supremum risk is attained at an extreme point:
\[
  \sup_{P \in \mathcal{P}} L_P(h) = \max_{P \in \mathrm{ex}(\mathcal{P})} L_P(h).
\]
By construction, $\mathcal{P}$ is the convex hull of the finite collection $\{P_{ij}\}$, meaning its extreme points are necessarily a subset of these generators ($\mathrm{ex}(\mathcal{P}) \subseteq \{P_{ij}\}$). Consequently, evaluating the maximum over the extreme points is equivalent to maximizing over the discrete set of indices $(i,j) \in [N_X] \times [N_Y]$. 

Substituting this finite inner maximum back into the original outer minimization over $\mathcal{H}$ directly yields the discrete minimax formulation in \eqref{eq:dro-minmax}.
\end{proof}

\section{Additional Empirical Analysis, Ablation Studies, Discussion, Limitations, and Future Works}~\label{sec:app:results}

\subsection{Empirical Analysis and Ablation Studies}

\paragraph{~\cref{prp:general-bounds} (General Distance Bounds.)}
We simulated a structured credal set by defining $20$ Gaussian covariate distributions $P_X$, considering $15$ distributions with standard deviation $\sigma = 1$ and means placed on a uniform grid over $[-3, 3]$, together with $5$ additional randomized Gaussians. We further defined $10$ labeling mechanisms. Soft labeling mechanisms were specified via conditional distributions $P_{Y|X}$ parameterized by sigmoid link functions with slope $1.0$, while deterministic hard labeling mechanisms were implemented using threshold classifiers of the form $\mathbf{1}(x > \theta)$. Thresholds (and sigmoid biases) were selected from a grid over $[-4, 4]$. This construction yielded a total of $20 \times 20 \times 10 \times 10 = 40{,}000$ distribution pairs for evaluation. Table~\ref{tab:gen-dist-bounds} presents a detailed breakdown of the tightness of the bounds across various settings. Across all evaluated cases, the empirical total variation distances satisfied the theoretical bounds with essentially zero slack (up to numerical precision) in the pure environment-shift and pure labeler-shift regimes and only moderate slack in the simultaneous joint-shift setting. 



\color{black}

\paragraph{Disagreement-Diameter Characterization.}

We validated the observable diameter characterizations across both the deterministic hard labeling regime (\cref{eq:diameter-is-disagreement}), and the soft label regime (\ref{eq:diameter-is-l1}).  We sampled $500$ random Gaussian distributions with varying means and variances and applied two distinct labeling mechanisms: (1) a threshold classifier (deterministic), (2) Probit link functions with $\kappa \in \{1,3\}$. Across both settings, the empirical total (mean disagreement rate or mean $L_1$ distance) proved to be unbiased and tightly concentrated around the theoretical credal diameter. In the soft label regime, the mean estimation gap was negligible ($3.65 \times 10^{-5}$), and deviations from the ground truth diameter remained within $0.05$ for $99.9\%$ of the trials. Together, these results demonstrate that the diameter of our structured credal set is fully observable in the fixed covariate regime and collapses the abstract geometric and measure-theoretic concept to simply measuring disagreements. The detailed metrics and quantiles are shown in Table~\ref{tab:gen-dist-bounds}. Finally, we also evaluated the noisy labeling mechanism framework (\cref{eq:diam-is-noise} where the labeling mechanisms flipped the ground truth. According to Table \ref{tab:noise-ablation}, when allowing noise rates for the labeling mechanisms up to $\epsilon=0.5$, the estimator remained robust, recovering the theoretical diameter with negligible bias $\approx 0.02$, confirming that the diameter of the credal set can be accurately observed even when the labeling mechanisms are imperfect.

\paragraph{\cref{thm:concentration} (Sample Complexity).}
The empirical results also provide a strong validation for the concentration bounds derived in \cref{thm:concentration} across both deterministic hard labels and probabilistic soft labels. Over three covariate distributions ($\mathcal{N}(0,1), \mathcal{N}(2,1), \mathcal{N}(0,2)$ in the hard-label regime, the empirical diameter converges at the theoretically consistent rate $O(\frac{1}{\sqrt{n}})$ with negligible violation probability and the tightness ratio between the theoretical bound and the $95^{th}$ percentile error remains stable around $\approx 2$ indicating a practically informative nound at moderate sample sizes. In the soft label regime, the convergence rate for a probabilistic labeling mechanism using both the probit link function and the sigmoid function for the covariate distribution $\mathcal{N}(0,1)$ was theoretically consistent, and violations were absent. Moreover, in the soft label regime, the mean empirical diameter already tracks close to the population value at small $n$, reflecting the lower variance signal induced by observing the probabilities themselves.

\paragraph{\cref{thm:concentration} (Labeling Mechanism Complexity).} We evaluate the stability of the diameter estimation as the number of labeling mechanisms $N_Y$ scales from $2$ to $1000$. Since the true diameter of the structured credal set in the fixed covariate regime $N_X=1$ is given by $\eta_{Y|X}^\ast = \mathbb{E}_{X \sim P_X} \left[ d_{\mathrm{TV}}\big(P_{Y|X}^{(j)}, P_{Y|X}^{(j')}\big) \right]$, the \emph{gating effect} plays an important role in determining the diameter of the credal set. To rigorously isolate the impact of mechanism complexity, we constructed a framework where the feature space $\mathcal{X}$ is partitioned into disjoint blocks defined by the quantiles of the marginal distribution $P_X$. Each labeling mechanism is restricted to a specific block, ensuring that any pair of mechanisms from different blocks has disjoint supports. Therefore, any pair of labeling mechanisms across blocks could be the maximal disagreeing pair. By fixing the probability mass of these blocks, we structurally pin the $\eta_{Y|X}^\ast$ at a constant value, preventing it from growing with $N_Y$. Since the blocks are defined via quantiles, this construction renders error bounds and theoretical properties invariant to the specific $P_X$, needing us to only evaluate the effects on a single feature distribution. Additionally, to also consider the possibility that adding a new labeling mechanism might also include an "outlier," we evaluated the mechanism complexity where $\eta_{Y|X}^\ast$ expands with the number of labeling mechanisms as well. Under both scenarios, the bounds remain consistently non-vacuous, with estimation error scaling logarithmically in $N_Y$. Crucially, the tightness ratio stabilizes between $1.4$ and $3.3$ even up to $N_Y=1000$, confirming that the framework scales robustly to a high number of labeling mechanisms without degradation.

\subsection{Discussion}

Tables \ref{tab:gen-dist-bounds},\ref{tab:det_diameter-ablation}, \ref{tab:softlabel-kappa-ablation}, \ref{tab:noise-ablation} collectively validate both the tightness of the proposed theoretical bounds and the practical observability of the structured credal diameter. Table \ref{tab:gen-dist-bounds} shows that the general distance bounds in Theorem \ref{prp:general-bounds} achieve 100\% empirical coverage across all evaluated distribution pairs, with zero gaps in pure covariate or pure label shift regimes and moderate slack only in joint-shift settings, confirming correctness and robustness of the decomposition 
While Table \ref{tab:det_diameter-ablation} demonstrates that, in the deterministic hard-label regime, the empirical estimator of the diameter is essentially unbiased, with mean absolute errors on the order of $10^-2$ and empirical and analytic diameters matching up to three decimal places, indicating fast concentration. Table \ref{tab:softlabel-kappa-ablation} extends this result to soft-label settings, showing similarly small estimation gaps and negligible signed bias across different probit sharpness levels, thereby confirming the accuracy of the L1-based characterization of the diameter. Finally, Table \ref{tab:noise-ablation} verifies the noisy-annotator model, where empirical diameters closely track the closed-form theoretical values across increasing noise budgets, with low RMSE and controlled slack, demonstrating robustness of the framework under aleatoric label noise.

\begin{table}[t]

\caption{Ablation study characterizing the tightness of the general distance bounds in \cref{prp:general-bounds}.
We evaluate $20$ Gaussian covariates and $10$ labeling mechanisms across all pair classes.
The lower-gap and upper-gap are defined as
$\epsilon_{\mathrm{LB}} := d_{\mathrm{TV}}(P_{ij}, P_{i'j'}) - \mathrm{LB}(P_{ij}, P_{i'j'})$
and
$\epsilon_{\mathrm{UB}} := \mathrm{UB}(P_{ij}, P_{i'j'}) - d_{\mathrm{TV}}(P_{ij}, P_{i'j'})$,
respectively.
Entries report the mean gap with $\,[\min, \max]$ over all pairs within each class. All evaluated pairs satisfy the bounds, resulting in a $100\%$ empirical coverage. Slack values on the order of $10^{-8}$ or smaller (numerical precision) and are reported as 0.}
\label{tab:gen-dist-bounds}
\centering
\small
\begin{tabular}{l r c c}
\toprule
Pair class & \#Pairs & $\epsilon_{\mathrm{LB}}$ (mean [min,max]) & $\epsilon_{\mathrm{UB}}$ (mean [min,max]) \\
\midrule
\multicolumn{4}{c}{\textbf{Soft labels (Sigmoid)}} \\
\midrule
Joint Shift & 34,200 & 0.242 [3.95e-04, 0.96] & 0.165 [2.81e-04, 0.803] \\
Fixed Covariate (Conditional Shift only) & 1,800 & 0 [0, 0] & 0 [0, 0] \\
Fixed Labeling Mechanism (Covariate Shift only) & 3,800 & 0 [0, 0] & 0 [0, 0] \\
Identical pair & 200 & 0 [0, 0] & 0 [0, 0] \\
\midrule
\multicolumn{4}{c}{\textbf{Hard labels (Thresholds)}} \\
\midrule
Joint Shift & 34,200 & 0.228 [0, 1] & 0.096 [0, 0.936] \\
Fixed Covariate (Conditional Shift only) & 1,800 & 0 [0, 0] & 0 [0, 0] \\
Fixed Labeling Mechanism (Covariate Shift only) & 3,800 & 0 [0, 0] & 0 [0, 0] \\
Identical pair & 200 & 0 [0, 0] & 0 [0, 0] \\
\bottomrule
\end{tabular}
\end{table}

\begin{table}[!htbp]
\caption{
Finite-sample validation of the exact diameter characterization
in \cref{eq:diameter-is-disagreement} for deterministic labeling mechanisms.
Across $500$ Gaussian covariate distributions and $5$ fixed threshold classifiers, the empirical diameter is estimated from $n=1000$ samples per environment,
repeated over $100$ runs.
Reported metrics summarize the gap between empirical and analytic diameters;
confidence intervals are bootstrap-based.
Errors are small, symmetric, and centered at zero, confirming unbiasedness
and rapid concentration of the empirical estimator.
}\label{tab:det_diameter-ablation}
\centering
\small
\begin{tabular}{l c c}
\toprule
Metric & Estimate & 95\% CI \\
\midrule
Mean absolute gap & 0.0117 & [0.0116, 0.0117] \\
Median absolute gap & 0.0097 & [0.0096, 0.0098] \\
RMSE (gap) & 0.0148 & [0.0147, 0.0149] \\
Mean gap (signed) & 8.98e-06 & [-0.0001, 0.0002] \\
90\% abs. gap quantile & 0.0244 & [0.0242, 0.0246] \\
95\% abs. gap quantile & 0.0293 & [0.029, 0.0295] \\
99\% abs. gap quantile & 0.0384 & [0.038, 0.0389] \\
Mean analytic diameter & 0.465 & \textemdash \\
Mean empirical diameter & 0.465 & [0.4648, 0.4651] \\
\bottomrule
\end{tabular}
\end{table}
\begin{table}[!htbp]
\caption{
Finite-sample validation of the soft-label diameter characterization in \cref{eq:diameter-is-l1} under probit labeling mechanisms.
We compare three probit sharpness settings $\kappa \in \{1,2,3\}$ (with the same Gaussian covariate setup and fixed labelers across runs).
Reported metrics summarize the gap between the empirical diameter (estimated from finite samples) and the analytic diameter, including mean/median absolute gap, RMSE, and tail quantiles of the absolute gap; brackets denote $95\%$ bootstrap confidence intervals.
Mean signed gaps are near zero, indicating negligible bias.
}\label{tab:softlabel-kappa-ablation}
\centering
\small
\begin{tabular}{l c c c}
\toprule
Metric & $\kappa=1$ (est. [CI]) & $\kappa=2$ (est. [CI]) & $\kappa=3$ (est. [CI]) \\
\midrule
Mean absolute gap & 0.0084 [0.0084, 0.0085] & 0.0101 [0.01, 0.0101] & 0.0106 [0.0105, 0.0107] \\
Median absolute gap & 0.007 [0.0069, 0.007] & 0.0084 [0.0083, 0.0084] & 0.0088 [0.0087, 0.0089] \\
RMSE (gap) & 0.0107 [0.0107, 0.0108] & 0.0128 [0.0127, 0.0129] & 0.0135 [0.0134, 0.0135] \\
Mean gap (signed) & 3.60e-05 [-7.00e-05, 0.0001] & 3.63e-05 [-9.37e-05, 0.0002] & 3.65e-05 [-9.91e-05, 0.0002] \\
90\% abs. gap quantile & 0.0176 [0.0175, 0.0178] & 0.0211 [0.0209, 0.0212] & 0.0222 [0.022, 0.0223] \\
95\% abs. gap quantile & 0.0212 [0.0211, 0.0214] & 0.0253 [0.0251, 0.0255] & 0.0265 [0.0263, 0.0267] \\
99\% abs. gap quantile & 0.0281 [0.0278, 0.0284] & 0.0335 [0.0332, 0.0339] & 0.035 [0.0346, 0.0354] \\
Mean analytic diameter & 0.4455 (\textemdash) & 0.4599 (\textemdash) & 0.4627 (\textemdash) \\
Mean empirical diameter & 0.4456 [0.4455, 0.4457] & 0.4599 [0.4598, 0.4601] & 0.4627 [0.4626, 0.4629] \\
\bottomrule
\end{tabular}
\end{table}
\begin{table}[!htbp]
\caption{Noisy-annotator ablation for Eq.~(\ref{eq:diam-is-noise}) across noise budgets $\varepsilon_{\max}\in\{0.1,0.25,0.5\}$. Entries report estimate [95\% bootstrap CI] for the true diameter, empirical diameter, and error/slack metrics.}

\label{tab:noise-ablation}
\centering
\small
\begin{tabular}{l c c c}
\toprule
Metric & $\varepsilon_{\max}=0.1$ & $\varepsilon_{\max}=0.25$ & $\varepsilon_{\max}=0.5$ \\
\midrule
True diameter (closed form) & 0.1764 (\textemdash) & 0.3694 (\textemdash) & 0.4998 (\textemdash) \\
Mean empirical diameter & 0.1783 [0.1782, 0.1783] & 0.3716 [0.3715, 0.3717] & 0.5212 [0.5211, 0.5212] \\
Upper bound & 0.1772 (\textemdash) & 0.3706 (\textemdash) & 0.4998 (\textemdash) \\
Mean gap to true (signed) & 0.0018 [0.0017, 0.0019] & 0.0022 [0.0021, 0.0023] & 0.0214 [0.0213, 0.0215] \\
Mean abs. gap to true & 0.0087 [0.0086, 0.0088] & 0.0108 [0.0108, 0.0109] & 0.0214 [0.0213, 0.0215] \\
RMSE to true & 0.0109 [0.0109, 0.011] & 0.0136 [0.0136, 0.0137] & 0.0235 [0.0234, 0.0235] \\
95\% abs. gap quantile (to true) & 0.0214 [0.0212, 0.0216] & 0.0267 [0.0265, 0.0269] & 0.0381 [0.0379, 0.0384] \\
99\% abs. gap quantile (to true) & 0.0283 [0.0279, 0.0286] & 0.0353 [0.0348, 0.0357] & 0.0456 [0.0452, 0.046] \\
Upper-bound violation rate & 0.5162 [0.5123, 0.5202] & 0.5167 [0.5125, 0.5208] & 0.9935 [0.9927, 0.9942] \\
Mean slack to upper (signed) & 0.001 [0.0009, 0.0011] & 0.001 [0.0009, 0.0011] & 0.0213 [0.0212, 0.0214] \\
Mean abs. slack to upper & 0.0086 [0.0086, 0.0087] & 0.0107 [0.0107, 0.0108] & 0.0213 [0.0213, 0.0214] \\
\bottomrule
\end{tabular}
\end{table}

Tables \ref{tab:det_n01_samp_complexity}, \ref{tab:det_n02_samp_complexity}, \ref{tab:det_n21_samp_complexity}, and \ref{tab:iv21} empirically validate the predicted sample complexity behavior of the diameter estimator in the deterministic hard-label regime across different covariate distributions. As the sample size increases, the median and high-quantile absolute estimation errors consistently decrease at a rate close to $O(n^{-1/2})$, confirming the theoretical concentration guarantees. Across all settings, the empirical violation probability remains effectively zero, indicating that the Hoeffding-based bounds are conservative yet reliable. Moreover, the tightness ratios remain stable and bounded (approximately between 1.7 and 2.5), demonstrating that the proposed concentration bounds are not only asymptotically correct but also practically informative across varying population diameters and covariate shifts.

\begin{table}[!htbp]
\centering
\setlength{\tabcolsep}{4pt}
\renewcommand{\arraystretch}{1.1}
\caption{Sample complexity concentration bounds in the deterministic hard label regime under covariate distribution N(0,1). Population diameter $\eta_{Y|X}^\ast = 0.683$. Repetitions $R = 2000$. Here, $n$ denotes sample size. $\widehat{p}_{\mathrm{viol}}$ is the estimated probability of violations. Let $\epsilon := \left|\eta_{Y|X}^\ast - \hat{\eta}_{Y|X}\right|$ and $q_{\alpha}(\epsilon)$ denote the empirical $\alpha$-quantile of $\epsilon$.}
\label{tab:det_n01_samp_complexity}
\begin{adjustbox}{max width=\textwidth}
\begin{tabular}{rrrrrrrr}
\toprule
n & $q_{0.5}(\epsilon)$ & $q_{0.95}(\epsilon)$ & $\epsilon_{\mathrm{Hoeff}}^{\cup}$ & $\widehat{p}_{\mathrm{viol}}$ & $\mathrm{CI}^{\mathrm{Wilson}}_{0.95,\mathrm{low}}$ & $\mathrm{CI}^{\mathrm{Wilson}}_{0.95,\mathrm{high}}$ & Tightness Ratio \\
\midrule
10 & 0.117 & 0.283 & 0.547 & 0.000 & 0.000 & 0.002 & 1.936 \\
30 & 0.051 & 0.183 & 0.316 & 0.000 & 0.000 & 0.002 & 1.730 \\
100 & 0.033 & 0.093 & 0.173 & 0.000 & 0.000 & 0.002 & 1.867 \\
500 & 0.013 & 0.039 & 0.077 & 0.000 & 0.000 & 0.002 & 1.969 \\
1000 & 0.010 & 0.030 & 0.055 & 0.000 & 0.000 & 0.002 & 1.844 \\
2000 & 0.007 & 0.020 & 0.039 & 0.000 & 0.000 & 0.002 & 1.954 \\
5000 & 0.004 & 0.013 & 0.024 & 0.000 & 0.000 & 0.002 & 1.870 \\
10000 & 0.003 & 0.009 & 0.017 & 0.000 & 0.000 & 0.002 & 1.843 \\
20000 & 0.002 & 0.006 & 0.012 & 0.000 & 0.000 & 0.002 & 1.970 \\
100000 & 0.001 & 0.003 & 0.005 & 0.000 & 0.000 & 0.002 & 1.855 \\
\bottomrule
\end{tabular}

\end{adjustbox}
\end{table}

\begin{table}[!htbp]
\centering
\setlength{\tabcolsep}{4pt}
\renewcommand{\arraystretch}{1.1}
\caption{Sample complexity concentration bounds in the deterministic hard label regime under covariate distribution N(0,2). Population diameter $\eta_{Y|X}^\ast = 0.383$. Repetitions $R = 2000$. Here, $n$ denotes sample size. $\widehat{p}_{\mathrm{viol}}$ is the estimated probability of violations. Let $\epsilon := \left|\eta_{Y|X}^\ast - \hat{\eta}_{Y|X}\right|$ and $q_{\alpha}(\epsilon)$ denote the empirical $\alpha$-quantile of $\epsilon$.}
\label{tab:det_n02_samp_complexity}
\begin{adjustbox}{max width=\textwidth}
\begin{tabular}{rrrrrrrr}
\toprule
n & $q_{0.5}(\epsilon)$ & $q_{0.95}(\epsilon)$ & $\epsilon_{\mathrm{Hoeff}}^{\cup}$ & $\widehat{p}_{\mathrm{viol}}$ & $\mathrm{CI}^{\mathrm{Wilson}}_{0.95,\mathrm{low}}$ & $\mathrm{CI}^{\mathrm{Wilson}}_{0.95,\mathrm{high}}$ & Tightness Ratio \\
\midrule
10 & 0.117 & 0.283 & 0.547 & 0.000 & 0.000 & 0.002 & 1.935 \\
30 & 0.050 & 0.183 & 0.316 & 0.000 & 0.000 & 0.002 & 1.727 \\
100 & 0.033 & 0.097 & 0.173 & 0.001 & 0.000 & 0.003 & 1.783 \\
500 & 0.015 & 0.043 & 0.077 & 0.000 & 0.000 & 0.002 & 1.803 \\
1000 & 0.011 & 0.030 & 0.055 & 0.001 & 0.000 & 0.004 & 1.820 \\
2000 & 0.007 & 0.021 & 0.039 & 0.000 & 0.000 & 0.002 & 1.836 \\
5000 & 0.005 & 0.013 & 0.024 & 0.000 & 0.000 & 0.002 & 1.816 \\
10000 & 0.003 & 0.009 & 0.017 & 0.000 & 0.000 & 0.002 & 1.836 \\
20000 & 0.002 & 0.007 & 0.012 & 0.001 & 0.000 & 0.003 & 1.861 \\
100000 & 0.001 & 0.003 & 0.005 & 0.001 & 0.000 & 0.003 & 1.719 \\
\bottomrule
\end{tabular}

\end{adjustbox}
\end{table}

\begin{table}[!htbp]
\centering
\setlength{\tabcolsep}{4pt}
\renewcommand{\arraystretch}{1.1}
\caption{Sample complexity concentration bounds in the deterministic hard label regime under covariate distribution N(2,1). Population diameter $\eta_{Y|X}^\ast = 0.157$. Repetitions $R = 2000$. Here, $n$ denotes sample size. $\widehat{p}_{\mathrm{viol}}$ is the estimated probability of violations. Let $\epsilon := \left|\eta_{Y|X}^\ast - \hat{\eta}_{Y|X}\right|$ and $q_{\alpha}(\epsilon)$ denote the empirical $\alpha$-quantile of $\epsilon$.}
\label{tab:det_n21_samp_complexity}
\begin{adjustbox}{max width=\textwidth}
\begin{tabular}{rrrrrrrr}
\toprule
n & $q_{0.5}(\epsilon)$ & $q_{0.95}(\epsilon)$ & $\epsilon_{\mathrm{Hoeff}}^{\cup}$ & $\widehat{p}_{\mathrm{viol}}$ & $\mathrm{CI}^{\mathrm{Wilson}}_{0.95,\mathrm{low}}$ & $\mathrm{CI}^{\mathrm{Wilson}}_{0.95,\mathrm{high}}$ & Tightness Ratio \\
\midrule
10 & 0.057 & 0.243 & 0.547 & 0.000 & 0.000 & 0.002 & 2.255 \\
30 & 0.043 & 0.143 & 0.316 & 0.000 & 0.000 & 0.002 & 2.215 \\
100 & 0.027 & 0.077 & 0.173 & 0.000 & 0.000 & 0.002 & 2.239 \\
500 & 0.011 & 0.031 & 0.077 & 0.000 & 0.000 & 0.002 & 2.473 \\
1000 & 0.008 & 0.023 & 0.055 & 0.000 & 0.000 & 0.002 & 2.349 \\
2000 & 0.005 & 0.016 & 0.039 & 0.000 & 0.000 & 0.002 & 2.449 \\
5000 & 0.004 & 0.010 & 0.024 & 0.000 & 0.000 & 0.002 & 2.332 \\
10000 & 0.002 & 0.007 & 0.017 & 0.000 & 0.000 & 0.002 & 2.402 \\
20000 & 0.002 & 0.005 & 0.012 & 0.000 & 0.000 & 0.002 & 2.474 \\
100000 & 0.001 & 0.002 & 0.005 & 0.000 & 0.000 & 0.002 & 2.406 \\
\bottomrule
\end{tabular}

\end{adjustbox}
\end{table}

\begin{table}[!htbp]
\centering
\setlength{\tabcolsep}{4pt}
\renewcommand{\arraystretch}{1.1}
\caption{Labeling mechanism complexity concentration bounds in the deterministic hard label regime (Interval Method) under covariate distribution N(2,1). Fixed sample size $n = 1000$. Repetitions $R = 2000$. Here, $N_Y$ denotes the number of labeling mechanisms. $\widehat{p}_{\mathrm{viol}}$ is the estimated probability of violations. Let $\epsilon := \left|\eta_{Y|X}^\ast - \hat{\eta}_{Y|X}\right|$ and $q_{\alpha}(\epsilon)$ denote the empirical $\alpha$-quantile of $\epsilon$. The union Hoeffding bound is $\epsilon_{\mathrm{Hoeff}}^{\cup}$.}
\label{tab:iv21}
\begin{adjustbox}{max width=\textwidth}
\begin{tabular}{rrrrrrrr}
\toprule
$N_Y$ & $q_{0.5}(\epsilon)$ & $q_{0.95}(\epsilon)$ & $\epsilon_{\mathrm{Hoeff}}^{\cup}$ & $\widehat{p}_{\mathrm{viol}}$ & $\mathrm{CI}^{\mathrm{Wilson}}_{0.95,\mathrm{low}}$ & $\mathrm{CI}^{\mathrm{Wilson}}_{0.95,\mathrm{high}}$ & Tightness Ratio \\
\midrule
2 & 0.003 & 0.009 & 0.043 & 0.000 & 0.000 & 0.002 & 4.605 \\
5 & 0.007 & 0.022 & 0.055 & 0.000 & 0.000 & 0.002 & 2.496 \\
12 & 0.008 & 0.022 & 0.063 & 0.000 & 0.000 & 0.002 & 2.820 \\
20 & 0.008 & 0.022 & 0.067 & 0.000 & 0.000 & 0.002 & 3.004 \\
30 & 0.008 & 0.022 & 0.070 & 0.000 & 0.000 & 0.002 & 3.226 \\
50 & 0.008 & 0.022 & 0.073 & 0.000 & 0.000 & 0.002 & 3.336 \\
80 & 0.008 & 0.022 & 0.077 & 0.000 & 0.000 & 0.002 & 3.409 \\
100 & 0.007 & 0.023 & 0.078 & 0.000 & 0.000 & 0.002 & 3.328 \\
200 & 0.007 & 0.022 & 0.082 & 0.000 & 0.000 & 0.002 & 3.825 \\
500 & 0.008 & 0.022 & 0.088 & 0.000 & 0.000 & 0.002 & 3.986 \\
1000 & 0.008 & 0.023 & 0.092 & 0.000 & 0.000 & 0.002 & 3.991 \\
\bottomrule
\end{tabular}

\end{adjustbox}
\end{table}

\begin{table}[!htbp]
\centering
\setlength{\tabcolsep}{4pt}
\renewcommand{\arraystretch}{1.1}
\caption{Sample complexity concentration bounds in the soft-label regime using the Probit link function under the covariate distribution N(0,1). Population diameter $\eta_{Y|X}^\ast = 0.520$. Repetitions $R = 100$. Here, $n$ denotes sample size. $\widehat{p}_{\mathrm{viol}}$ is the estimated probability of violations. Let $\epsilon := \left|\eta_{Y|X}^\ast - \hat{\eta}_{Y|X}\right|$ and $q_{\alpha}(\epsilon)$ denote the empirical $\alpha$-quantile of $\epsilon$.}
\label{tab:prob_samp_comp}
\begin{adjustbox}{max width=\textwidth}
\begin{tabular}{rrrrrrrr}
\toprule
n & $q_{0.5}(\epsilon)$ & $q_{0.95}(\epsilon)$ & $\epsilon_{\mathrm{Hoeff}}^{\cup}$ & $\widehat{p}_{\mathrm{viol}}$ & $\mathrm{CI}^{\mathrm{Wilson}}_{0.95,\mathrm{low}}$ & $\mathrm{CI}^{\mathrm{Wilson}}_{0.95,\mathrm{high}}$ & Tightness Ratio \\
\midrule
10 & 0.038 & 0.104 & 0.547 & 0.000 & 0.000 & 0.037 & 5.269 \\
30 & 0.022 & 0.057 & 0.316 & 0.000 & 0.000 & 0.037 & 5.514 \\
100 & 0.011 & 0.032 & 0.173 & 0.000 & 0.000 & 0.037 & 5.420 \\
500 & 0.005 & 0.012 & 0.077 & 0.000 & 0.000 & 0.037 & 6.399 \\
1000 & 0.003 & 0.009 & 0.055 & 0.000 & 0.000 & 0.037 & 6.034 \\
2000 & 0.002 & 0.006 & 0.039 & 0.000 & 0.000 & 0.037 & 6.070 \\
10000 & 0.001 & 0.003 & 0.017 & 0.000 & 0.000 & 0.037 & 6.162 \\
20000 & 0.001 & 0.002 & 0.012 & 0.000 & 0.000 & 0.037 & 6.289 \\
100000 & 0.000 & 0.001 & 0.005 & 0.000 & 0.000 & 0.037 & 5.559 \\
\bottomrule
\end{tabular}

\end{adjustbox}
\end{table}

\begin{figure}[h]
     \centering
     \begin{subfigure}[h]{0.45\textwidth}
         \centering
          \includegraphics[width=\linewidth]{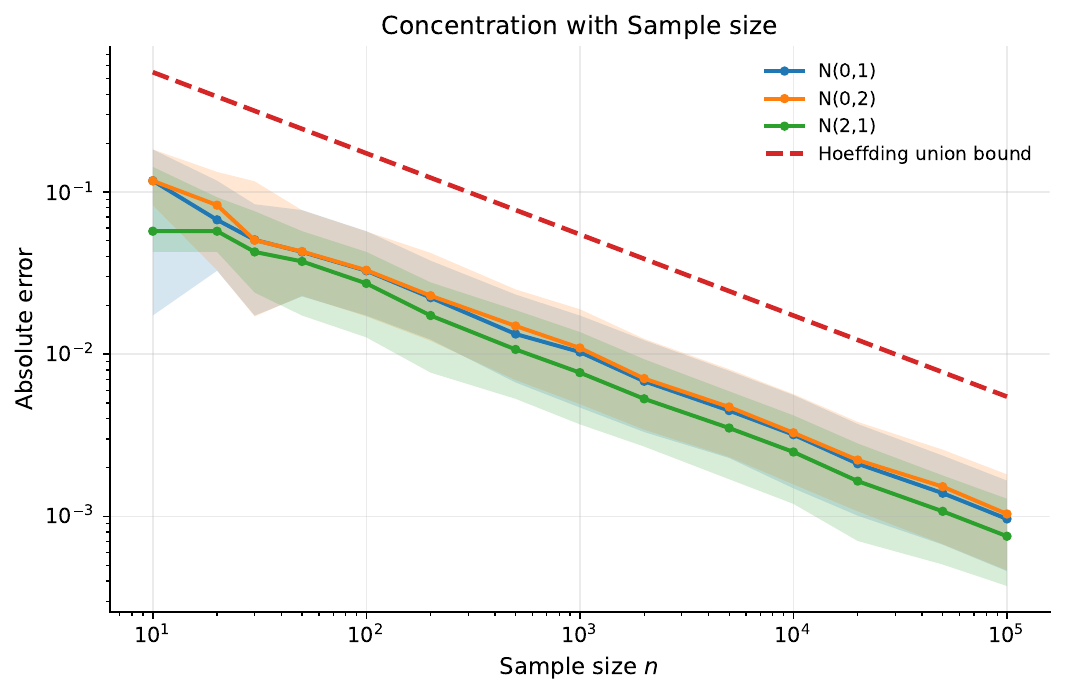}
         \caption{\footnotesize \textbf{Deterministic hard labels.} Absolute estimation error versus sample size (log–log scale) for Gaussian inputs with varying variances. The curves demonstrate the canonical $O(n^{-1/2})$ concentration rate, with empirical errors lying well below the Hoeffding union bound.}
         \label{fig:app:abalation-det-hard-lab-samp-complexity-combined}
     \end{subfigure}
     \hfill
     \begin{subfigure}[h]{0.45\textwidth}
         \centering
         \includegraphics[width=\linewidth]{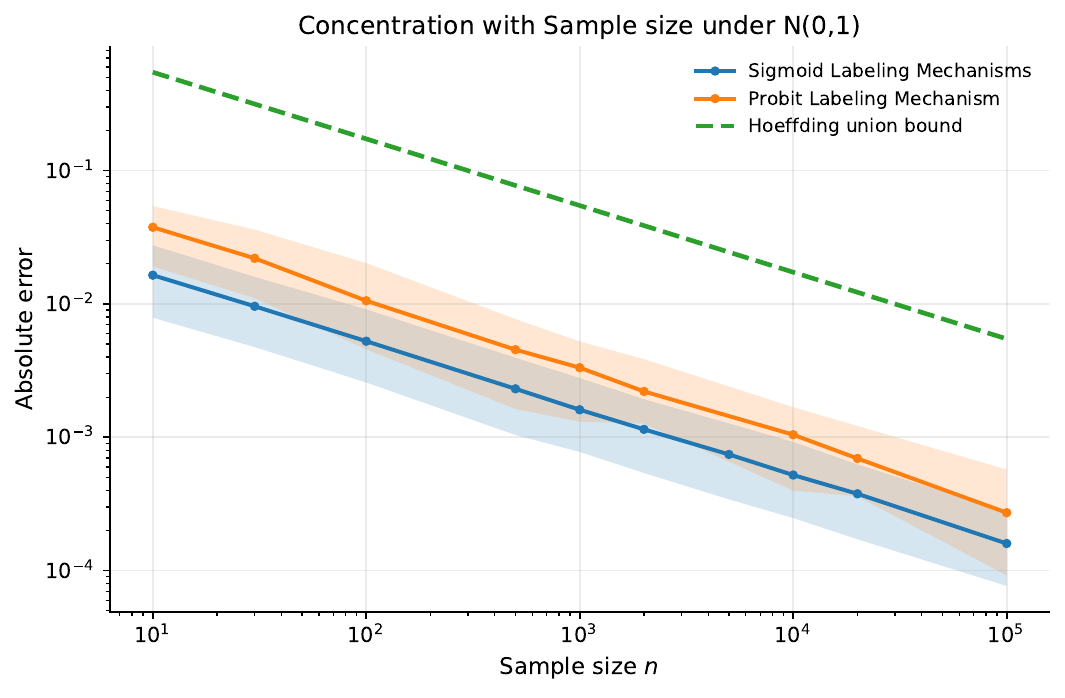}
         \caption{\footnotesize \textbf{Soft (probabilistic) label.} Absolute estimation error versus sample size (log–log scale) under sigmoid and probit labeling mechanisms for $N(0,1)$ inputs. Both mechanisms exhibit $O(n^{-1/2})$ convergence, with sigmoid labeling achieving systematically lower error.}
         \label{fig:app:abalation-soft-lab-samp-complexity-combined}
     \end{subfigure}
          \hfill
     \begin{subfigure}[h]{0.6\textwidth}
         \centering
         \includegraphics[width=\linewidth]{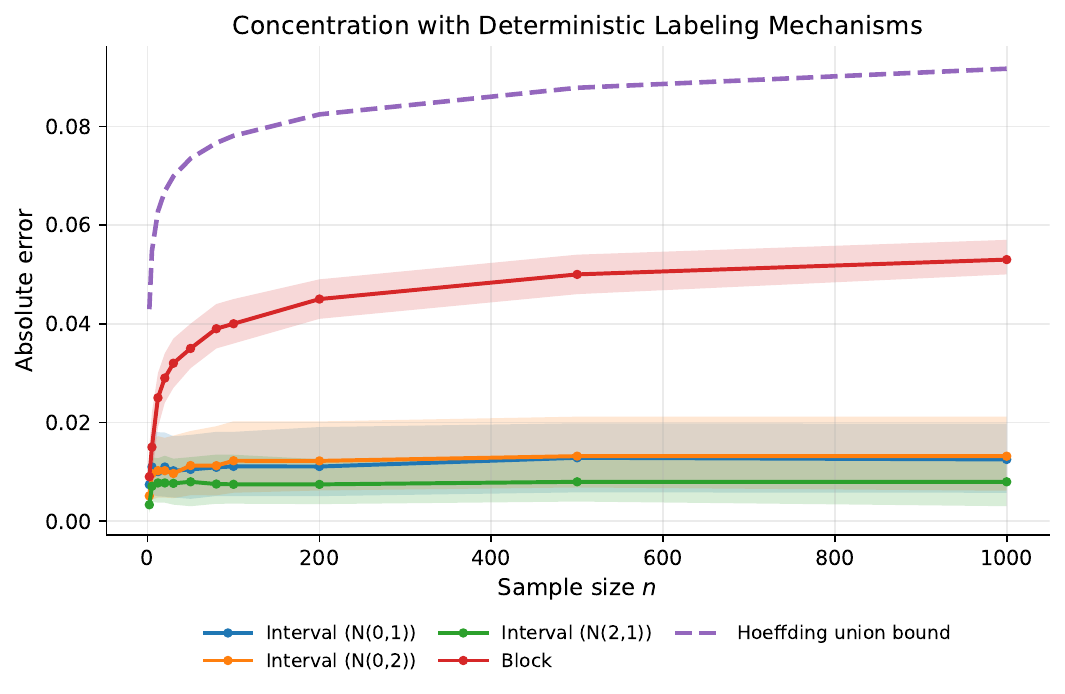}
         \caption{\footnotesize\textbf{Deterministic labeling mechanisms.} Absolute error as a function of sample size (linear scale) for interval-based and block-based deterministic labelers. Interval mechanisms yield stable, low error, while block labeling shows slower improvement and larger asymptotic error; the Hoeffding union bound provides a conservative upper envelope.}
         \label{fig:app:abalation-det-lab-mechanism-complexity}
     \end{subfigure}
     \caption{Empirical concentration behavior of estimation error under different data distributions and labeling mechanisms. Across settings, the observed error exhibits the expected theoretical scaling with sample size, and is compared against Hoeffding-type union bounds to highlight the gap between empirical performance and worst-case guarantees.}
     \label{fig:app:concentration}
\end{figure}

Table~\ref{tab:prob_samp_comp} obtains empirical concentration behavior in the soft-label regime with a Probit link under \(N(0,1)\) covariates. The median and 95th-percentile estimation errors decrease monotonically with sample size, exhibiting the predicted \(O(n^{-1/2})\) scaling. The Hoeffding-based bound remains conservative with zero observed violations across all trials, while the tightness ratio stays stable, indicating a consistent gap between worst-case guarantees and empirical performance.

Figure~\ref{fig:app:concentration} outlines the empirical concentration behavior of the proposed estimators under different labeling regimes. In Fig.~\ref{fig:app:abalation-det-hard-lab-samp-complexity-combined}, corresponding to deterministic hard labels, the absolute error decays approximately at the theoretical rate \(O(n^{-1/2})\) across all Gaussian input distributions, confirming that the estimator exhibits standard parametric concentration while remaining substantially below the Hoeffding union bound. Fig.~\ref{fig:app:abalation-soft-lab-samp-complexity-combined} highlights that soft labeling mechanisms, including sigmoid and probit models, preserve the same asymptotic scaling; however, the sigmoid mechanism consistently achieves lower error, suggesting improved variance properties due to smoother probabilistic supervision. Finally, Fig.~\ref{fig:app:abalation-det-lab-mechanism-complexity} demonstrates the impact of deterministic labeling structure: interval-based mechanisms maintain uniformly low error as the sample size increases, whereas block-based labeling converges more slowly and exhibits a higher error floor. Together, our results indicate that while the asymptotic rate is robust to the choice of labeling strategy, the constant factors and finite-sample behavior are strongly influenced by the labeling mechanism and data distribution.

\end{document}